\documentclass[11pt]{article}
\usepackage[utf8]{inputenc}
\usepackage{fullpage,times}

\usepackage{fullpage}
\usepackage{amsmath,amsfonts,amsthm,amssymb,multirow}
\usepackage{times}
\usepackage{graphicx}
\usepackage{floatpag}
\usepackage{algorithm}
\usepackage[noend]{algpseudocode}
\usepackage{bbold}
\usepackage{enumitem}
\usepackage{subcaption} 
\usepackage{calc}
\usepackage{tikz}
\usetikzlibrary{decorations.markings}
\tikzstyle{vertex}=[circle, draw, inner sep=0pt, minimum size=4pt, fill = black]

\usepackage{graphicx}
\usepackage{tabularx}
\makeatletter
\newcommand{\multiline}[1]{%
  \begin{tabularx}{\dimexpr\linewidth-\ALG@thistlm}[t]{@{}X@{}}
    #1
  \end{tabularx}
}
\makeatother
\usepackage{mathtools}
\DeclarePairedDelimiter\abs{\lvert}{\rvert}%

\makeatletter
\def\BState{\State\hskip-\ALG@thistlm}
\makeatother

\usepackage[compact]{titlesec}
\titlespacing{\section}{0pt}{3ex}{2ex}
\titlespacing{\subsection}{0pt}{2ex}{1ex}
\titlespacing{\subsubsection}{0pt}{0.5ex}{0ex}

\newtheorem{theorem}{Theorem}[section]

\newtheorem{corollary}{Corollary}[section]

\newtheorem{definition}{Definition}[section]
\newtheorem{lemma}{Lemma}[section]
\newtheorem{claim}{Claim}

\newtheorem{observation}{Observation}[section]

\makeatletter
\let\c@fconjecture\c@conjecture
\makeatother

\makeatletter
\let\c@fconj\c@conj
\makeatother

\def \eps {\varepsilon}

\newcommand{\ignore}[1]{}

\def\tO{\tilde{O}}

\title{Approximation Algorithms for Min-Distance Problems}
\author{Mina Dalirrooyfard, Virginia Vassilevska Williams, Nikhil Vyas,\\ Nicole Wein, Yinzhan Xu and Yuancheng Yu\thanks{minad@mit.edu, virgi@mit.edu, nvyas@mit.edu, nwein@mit.edu, xyzhan@mit.edu, ycyu@mit.edu, MIT EECS and CSAIL}}
\date{}

\begin{document}

\maketitle
\begin{abstract}
We study fundamental graph parameters such as the Diameter and Radius in directed graphs, when distances are measured using a somewhat unorthodox but natural measure: the distance between $u$ and $v$ is the {\em minimum} of the shortest path distances from $u$ to $v$ and from $v$ to $u$. The center node in a graph under this measure can for instance represent the optimal location for a hospital to ensure the fastest medical care for everyone, as one can either go to the hospital, or a doctor can be sent to help.

By computing All-Pairs Shortest Paths, all pairwise distances and thus the parameters we study can be computed exactly in $\tilde{O}(mn)$ time for directed graphs on $n$ vertices, $m$ edges and nonnegative edge weights. Furthermore, this time bound is tight under the Strong Exponential Time Hypothesis [Roditty-Vassilevska W. STOC 2013] so it is natural to study how well these parameters can be \emph{approximated} in $O(mn^{1-\eps})$ time for constant $\eps>0$. Abboud, Vassilevska Williams, and Wang [SODA 2016] gave a  polynomial factor approximation for Diameter and Radius, as well as a constant factor approximation for both problems in the special case where the graph is a DAG. We greatly improve upon these bounds by providing the first constant factor approximations for Diameter, Radius and the related Eccentricities problem in general graphs. Additionally, we provide a hierarchy of algorithms for Diameter that gives a time/accuracy trade-off.

%$O(mn^{1-\eps})$ time for constant $\eps>0$.  Such results were previously only known for Diameter and Radius under the min-distance in directed {\em acyclic} graphs [Abboud-Vassilevska W.-Wang SODA 2016].
\end{abstract}

\section{Introduction}
The diameter, radius and eccentricities of a graph are fundamental parameters that have been extensively studied  \cite{chung, Hakimi, ChepoiD94,eppstein-planar-jv, aingworth, Corneil01, Chepoi02, Dvir04, BenMoshe, BeKa07, WN08, Yuster10, Chan12, FHW12, WeYu13, RV13, ChechikLRSTW14, AbboudGW15, BCH+15} (and many others). The eccentricity of a vertex $v$ is the largest distance between $v$ and any other vertex. The diameter is the maximum eccentricity of a vertex in the graph, thus measuring how far apart two nodes can be, and the radius is the minimum eccentricity, measuring the maximum distance to the most central node. 

The distance between two vertices in an undirected graph is just the shortest path distance $d(\cdot,\cdot)$ between them. For directed graphs, however, this notion of distance $d$ is no longer necessarily symmetric, and rather than being a distance {\em between} two nodes, it measures the distance in a given direction. Several related notions of pairwise distance that are symmetric have been studied. These include the roundtrip distance \cite{CW99} which for two vertices $u$ and $v$ is just $d(u,v)+d(v,u)$, the max-distance \cite{AbboudWW16} which is $\max\{d(u,v),d(v,u)\}$, and the min-distance \cite{AbboudWW16} which is $\min\{d(u,v),d(v,u)\}$.

Each of these notions of distance has a particular application. For instance, one would have to pay the roundtrip distance when going to the store and back. On the other hand, if one needs medical assistance, one could either go to the hospital, or have a physician come to the home --- the time to receive care is then measured by the min-distance. Another example of min-distance is in symmetric-key encryption: any pair of parties can create a shared private key by using only one-way communication.

For each notion of distance, the diameter, radius and eccentricity parameters are well-defined. Given the shortest path distances $d(\cdot,\cdot)$ for all vertices, the parameters for each distance measure can be computed in $O(n^2)$ time in $n$ vertex graphs. The fastest known algorithms for All-Pairs Shortest Paths (APSP) \cite{ryanapsp,Pettie02,PettieR05} give the fastest known algorithms to compute these parameters exactly, running in $n^3/\exp(\sqrt{\log n})$ time and $O(mn+n^2\log\log n)$, respectively on $m$-edge, $n$-vertex graphs. 
Furthermore, under the Strong Exponential Hypothesis, there is no $O(m^{2-\eps})$ time algorithm for Diameter in unweighted graphs (and thus also for any of these notions of Diameter and Eccentricities in directed graphs)~\cite{RV13}. For Radius, the same lower bound holds but under the ``Hitting Set" conjecture \cite{AbboudWW16}. 

As exact computation is expensive, it makes sense to resort to approximation algorithms. For the shortest path distance versions of Diameter, Eccentricities and Radius, there are several fast algorithms that achieve various small constant approximation ratios~\cite{RV13,ChechikLRSTW14,CGR,BackursRSWW18}. For instance, for Diameter, a folklore linear time algorithm can achieve a $2$-approximation, and an $\tilde{O}(m^{3/2})$ time\footnote{We use $\tO$ notation to hide polylogarithmic factors} algorithm can achieve a $3/2$-approximation \cite{RV13,ChechikLRSTW14}.

Many of these algorithms \cite{RV13,ChechikLRSTW14,BackursRSWW18} work for any distance measure that satisfies the triangle inequality. Thus they work for the shortest paths distance, max-distance and roundtrip distance. The min-distance however does not satisfy the triangle inequality: e.g. you might have edges $(x,y)$ and $(z,y)$, and thus the min-distance between $x$ and $y$ and between $y$ and $z$ are both $1$, yet there may be no directed path between $x$ and $z$ in any direction, so that the min-distance between them may be $\infty$. 

This issue makes it much more difficult to design fast approximation algorithms for Min-Diameter, Min-Radius and Min-Eccentricities (the parameters of interest under the min-distance). The only known nontrivial algorithms are by Abboud et al.~\cite{AbboudWW16}. For Min-Diameter \cite{AbboudWW16} gives a near-linear time $2$-approximation algorithm if the input is a directed acyclic graph. For general graphs, the only nontrivial fast approximation algorithm is an $\tilde{O}(mn^{1-\eps})$ time $n^\eps$-approximation algorithm for any constant $\eps>0$.
(No constant factor approximation algorithm is known that runs significantly faster than just computing APSP.) For Min-Radius, \cite{AbboudWW16} gives an $\tilde{O}(m \sqrt n)$ time $3$-approximation algorithm for directed acyclic graphs. For general graphs, they only achieve a very weak $n$-approximation in near-linear time that checks if the Min-Radius is finite.
There are no known approximation algorithms for Min-Eccentricities faster than just computing APSP.

\subsection{Our Results}
%todo 
The main goal of our paper is to obtain new fast, $O(mn^{1-\eps})$ time for some constant $\eps>0$, algorithms for Min-Diameter, Min-Radius and Min-Eccentricities (thus beating the $\tilde{O}(mn)$ time of exact computation). We achieve this by developing powerful new techniques that can handle the complications that arise due to the fact that the min-distance does not satisfy the triangle inequality.

Our results are as follows.
For Min-Diameter we achieve a hierarchy of algorithms trading off running time with approximation accuracy. 

\begin{theorem}
%\label{thm:min_diam_param}
For any integer $0<\ell\le O(\log{n})$, there is an $\tilde{O}(mn^{1/(\ell+1)})$ time randomized algorithm that, given a directed weighted graph $G$ with edge weights non-negative and polynomial in $n$, can output an estimate $\tilde{D}$ such that $D / (4\ell-1) \le \tilde{D} \le D$ with high probability, where $D$ is the min-diameter of $G$. 
\end{theorem}

When we set $\ell=1$, we obtain an $\tilde{O}(m\sqrt{n})$ time $3$-approximation algorithm, and when we set $\ell=\lceil \log n\rceil$, we get an $\tilde{O}(m)$ time $O(\log n)$-approximation.

Our tradeoff achieves the first constant factor approximation algorithms for Min-Diameter in general graphs that run in $O(mn^{1-\eps})$ time for constant $\eps>0$. Such a result was only known for directed acyclic graphs, whereas for general graphs the only known efficient algorithm could achieve an $n^\eps$-approximation.

For Min-Radius, we also achieve the first constant factor approximation algorithm for general graphs running in $O(mn^{1-\eps})$ time for some constant $\eps>0$. Such a result was only known for directed acyclic graphs, whereas for general graphs the only known efficient algorithm could only check if the Min-Radius is finite.

\begin{theorem}
%\label{thm:radius}
For any constant $\delta$ with $1>\delta>0$, there is an $\tilde{O}(m\sqrt n/\delta)$ time randomized algorithm, that given a directed weighted graph $G$ with edge weights positive and polynomial in $n$, can output an estimate $R'$ such that $R\le R'\le (3+\delta)R$ with high probability, where $R$ is the min-radius of $G$. 
\end{theorem}

Finally, we obtain the first $O(mn^{1-\eps})$ time (for constant $\eps>0$) constant factor approximation algorithms for the Min-Eccentricities of all vertices in a graph. For unweighted graphs we are able to obtain a close to $3$ approximation in $\tilde{O}(m\sqrt n)$ time. For weighted graphs, our approximation factor grows to $5$, while the running time is the same. Previously, the only algorithm to approximate the Min-Eccentricities computed them exactly via an APSP computation.

\begin{theorem}
%\label{thm:eccntricities}
For any constant $\delta$ with $1>\delta>0$, there is an $\tilde{O}(m\sqrt n/\delta)$ time randomized algorithm, that given a directed weighted graph $G=(V,E)$ with weights positive and polynomial in $n$, can output an estimate $\eps'(s)$ for every vertex $s \in V$ such that $\eps(s)\le \eps'(s)\le (5+\delta)\eps(s)$ with high probability, where $\eps(s)$ is the min-eccentricity of vertex $s$ in $G$. 
\end{theorem}

\begin{theorem}
%\label{thm:eccntricities-unweighted}
For any constant $\delta$ with $1>\delta>0$, there is an $\tilde{O}(m\sqrt n/\delta^2)$ time randomized algorithm, that given a directed unweighted graph $G=(V,E)$, can output an estimate $\eps'(s)$ for every vertex $s \in V$ such that $\eps(s)\le \eps'(s)\le (3+\delta)\eps(s)$ with high probability, where $\eps(s)$ is the min-eccentricity of the vertex $s$ in $G$. 
\end{theorem}

\subsection{Our Techniques}\label{sec:tech}
To obtain our results, we develop powerful new techniques which we outline below.

\subparagraph*{Partial search graphs.}
The idea of partial search graphs is used in the algorithms of \cite{AbboudWW16} for Min-Radius and Min-Diameter on DAGs. These algorithms use the following high-level framework: perform Dijkstra's algorithm from some vertices and then perform a
\emph{partial} Dijkstra's algorithm from \emph{every} vertex. The partial search from a vertex $v$ is with respect to a carefully defined partial search graph $G_v\subset G$. The crux of the analysis for the algorithms on DAGs is to argue that if the executions of Dijkstra's algorithm on the full graph did not find a good estimate for the desired quantity (either min-diameter or min-radius), then the partial search from some vertex $v$ returns a good estimate of the min-eccentricity of $v$, which in turn is a good estimate for the desired quantity. 
In DAGs it is natural to define the partial search graphs $G_v$ by considering a topological ordering of the vertices and letting each $G_v$ be some interval containing $v$ (though defining the exact intervals requires some work). For general graphs it is completely unclear how to even define such intervals since there is no natural notion of an ordering of the vertices, and thus figuring out what the 
 $G_v$'s should be is nontrivial. 
%It would be ideal to reduce the problem on general graphs to the problem on DAGs, however it is unclear how to do this. Instead, 
Our approach to overcoming this hurdle is to carefully define a DAG-like structure in general graphs. Such a structure may be of independent interest.

\subparagraph{Defining a DAG-like structure in general graphs.}
It would be ideal to directly reduce the problem on general graphs to the problem on DAGs, however it is very unclear how to do this. Instead, we recognize that it suffices to define a \emph{DAG-like} structure in general graphs. As a first step, we use the following idea. Suppose we have performed Dijkstra's algorithm from a vertex $v$. We let $S_v=\{u:d(u,v)< d(v,u)\}$ and we let $T_v=\{u:d(u,v)> d(v,u)\}$\footnote{$u$'s with $d(u,v)=d(v,u)$ are added to either $S_v$ or $T_v$ as specified in the formal definition later}. Then, we partially order the vertices so that the vertices in $S_v$ appear before $v$ and those in $T_v$ appear after $v$. We note that this partial ordering is ``DAG-like" because it is consistent with the topological ordering of a DAG; that is, if we apply this partition into $S_v$ and $T_v$ to a DAG then there trivially exists a topological ordering such that every vertex in $S_v$ appears before $v$ and every vertex in $T_v$ appears after $v$. After partitioning into $S_v$ and $T_v$, we recursively partition each set to create a more precise partial ordering. Importantly, we show that by recursively sampling vertices randomly, we can guarantee that our partitioning is approximately balanced which is crucial for the runtime analysis.
 The obtained partial ordering is the starting point for all of our algorithms. 

\subparagraph*{Min-Diameter: graph augmentation.} 
%The Min-Diameter algorithm on DAGs from [todo cite] is as follows. Consider a topological ordering of the vertices in the graph and perform Dijkstra's algorithm from the middle vertex. Then recurse on the graphs induced by each half of the ordering (this recursion defines the previously referenced partial search graphs). The analysis of this algorithm relies on two properties: (1) if the true endpoints of the min-diameter fall into opposite halves of the ordering then the min-eccentricity of the middle vertex is a 2-approximation for the min-diameter, and (2) for all pairs of vertices in the same half of the ordering their min-distance in the graph induced by this half is the same as their min-distance in the full graph. For general graphs a simple argument shows that a property analogous to property 1 still holds: for any vertex $v$ if the true endpoints of the min-diameter fall into opposite sets $S_v$, $T_v$ then the min-eccentricity of $v$ is a 2-approximation for the min-diameter. 
The Min-Diameter algorithm on DAGs from \cite{AbboudWW16} relies heavily on the following key property of DAGs. Consider a topological ordering and the graphs induced by the first and second halves of the ordering; which are defined with respect to the middle vertex in the ordering. For all pairs of vertices in the same half of the ordering, their min-distance in the graph induced by this half is the same as their min-distance in the full graph. 
As previously mentioned, if we sample a vertex $v$, we can make sure that $S_v$ and $T_v$ are approximately balanced, so that we can think of $S_v$ and $T_v$ as corresponding to the first and second half of a DAG topological ordering, respectively.
%**As previously mentioned, $S_v$ and $T_v$ are analogous to the first and second halves (respectively) of the topological ordering of a DAG, 
However it is unclear how to obtain a property of $S_v$ and $T_v$ analogous to the above key property of DAGs. In particular, 
%However, it is unclear how to guarantee a property analogous to property 2 because 
the min-distance between a pair of vertices in the graph induced by $S_v$ could be wildly different from their min-distance in the full graph, since paths whose endpoints are in $S_v$ can contain vertices outside of $S_v$. To overcome this hurdle, 
%instead of simply recursing on the graphs induced by $S_v$ and $T_v$, 
we \emph{augment} the graph induced by $S_v$ and the graph induced by $T_v$ by carefully adding edges so that distances within these augmented graphs approximate distances in the original graph. 
%Then we can process the augmented version of $S_v$ and $T_v$ independently, either recursively or by brute-force.

\subparagraph*{Min-Radius: refined DAG-like structure}
%The Min-Radius algorithm on DAGs from [todo cite] is more involved than the Min-Diameter algorithm on DAGs, but the idea is as follows. To prune the set of potential min-centers, the algorithm defines intervals in the topological ordering that cannot contain the min-center.  Then it defines the partial search graph for each potential min-center to be a carefully defined interval in the ordering. The analysis of this algorithm relies on two properties: (1) the min-eccentricity of the min-center with respect to its partial search graph is at most its min-eccentricity with respect to the original graph, and (2) if a vertex has small min-eccentricity within its partial search graph then it also has small min-eccentricity with respect to the original graph. These two properties combined imply that the algorithm finds a vertex with small min-eccentricity within its partial search graph and thus small min-eccentricity with respect to the original graph.

Our Min-Radius algorithm is much more delicate than our Min-Diameter algorithm due to the fact that for Min-Radius we care about small distances instead of large distances. In particular, the graph augmentation idea from our Min-Diameter algorithm does not help for Min-Radius because although the augmentations do not distort large distances much, they heavily distort small distances. Furthermore, the previously mentioned DAG-like structure for general graphs does not suffice for Min-Radius. However we use it as a starting point to define a more refined DAG-like partial ordering.
%it the previously referenced DAG-like structure is not enough.  Since we care about small distances, becomes more delicate. 
%It is completely unclear how to define the partial search graphs in general graphs to guarantee either of the above properties. The augmented partial search graphs from the Min-Diameter algorithm do not help because although the augmentations do not distort large distances much, they heavily distort small distances, which is what we care about for Min-Radius. To overcome this problem, we refine our previously defined DAG-like partial ordering to a more detailed and delicate structure. 
Most of our algorithm is concerned with precisely arranging vertices in this partial ordering.
%so that we can define the partial search graphs with the desired properties.
%to be intervals in the ordering that guarantee the two desired properties above. 
Specifically, we structure the partial ordering to satisfy \emph{roughly} the following property: for every pair of vertices $u$, $v$ such that $u$ appears before $v$ in the partial ordering, $d(v,u)$ is large while $d(u,v)$ is small. 

\subsection{Notation}
Given a graph $G=(V,E)$, $n=|V|$ and $m=|E|$. Graphs are directed and have non-negative weights polynomial in $n$ unless otherwise specified. For any pair of vertices $u$ and $v$, the \emph{distance from $u$ to $v$} $d(u,v)$ is the length of the shortest directed path from $u$ to $v$. When the context is not clear, we write $d_G(u,v)$ to specify the graph $G$. The \emph{min-distance} between a pair of vertices $u$ and $v$ is $d_{min}(u,v)=\min\{d(u,v),d(v,u)\}$. The \emph{min-diameter} of a graph is $\max_{u,v\in V} d_{min}(u,v)$. The \emph{min-radius} of a graph is $\min_{v\in V} \max_{u\in V} d_{min}(u,v)$. For any vertex $v$, the \emph{min-eccentricity} of $v$ is $\eps(v)=\max_{u\in V}d_{min}(u,v)$. When the context is not clear, we say $\eps_G(v)$ to specify the graph $G$. Note that we do not use the \emph{min} subscript to denote the min-eccentricity of a vertex. For an algorithm with input size $n$ we use \emph{with high probability} to denote the probability $>1-1/n^c$ for all constants $c$. We say some quantity is $poly(n)$ to mean it is $O(n^c)$ for some fixed constant $c$. We use $\tO$ notation to hide polylogarithmic factors.

\subsection{Organization}
In Section~\ref{sec:overview} we give an overview of all of our algorithms, in Section~\ref{section:prelim} we describe a graph partitioning procedure that begins all of our algorithms, in Section~\ref{section:diam} we describe our Min-Diameter algorithms, in Section~\ref{section:rad} we describe our Min-Radius algorithm, and in Section~\ref{section:ecc} we describe our Min-Eccentricities algorithm.

\section{Overview of Algorithms}\label{sec:overview}

We use the algorithms from \cite{AbboudWW16} for Min-Diameter and Min-Radius on DAGs as inspiration. For each problem, we first outline the DAG algorithm and then provide intuition for how to apply these ideas to general graphs. %Our general framework is described in the previous section, and here we provide more detail.

%todo is it the previous section?
%todo make notation consistent with body of paper e.g. G_v
%todo define G[S] (graph induced by S)

\subsection{Min-Diameter}
\subsubsection*{Algorithm for DAGs}
We begin by outlining the $\tO(n+m)$ time 2-approximation algorithm for Min-Diameter on DAGs from \cite{AbboudWW16}. Consider a topological ordering of the vertices and perform Dijkstra's algorithm from the middle vertex $v$. Then recurse on the graphs induced by the vertices in the first half (before $v$) and in the second half (after $v$). A key observation in the analysis is that if the true endpoints $s^*$ and $t^*$ of the min-diameter fall on opposite sides of $v$ in the ordering, then the min-eccentricity $\eps(v)$ of $v$ is a 2-approximation for the min-diameter $D$. This is because if $\eps(v)<D/2$ and $s^*$ and $t^*$ fall on opposite sides of $v$ in the ordering, then $d(s^*,v)<D/2$ and $d(v,t^*)<D/2$ so $d(s^*,t^*)<D$, a contradiction. So, suppose (without loss of generality) that $s^*$ and $t^*$ both fall before $v$ in the ordering. Since the graph is a DAG, every path between $s^*$ and $t^*$ only uses vertices before $v$ in the ordering. Thus, the min-distance between $s^*$ and $t^*$ in the graph induced by the first half of the graph is still $D$.

\subsubsection*{Algorithm for general graphs}
We now outline a precursor to our Min-Diameter algorithm for general graphs that mimics the algorithm for DAGs. This $\tO(n+m)$ time algorithm does not achieve a constant approximation factor, however it provides intuition for our constant-factor approximation algorithms. We begin by  performing Dijkstra's algorithm from a vertex $v$ and constructing $S_v$ and $T_v$ as defined in the previous section. 
%Recall that the purpose of $S_v$ and $T_v$ is to define a partial ordering of the vertices, which defines a DAG-like structure. As described in the previous section,
Analogously to the DAG algorithm if the true min-diameter endpoints $s^*$ and $t^*$ fall into different sets $S_v$, $T_v$ then the min-eccentricity $\eps(v)$ is a 2-approximation. This is because if $\eps(v)<D/2$, $s^*\in S_v$, and $t^*\in T_v$ then $d(s^*,v)<D/2$ and $d(v,t^*)<D/2$ so $d(s^*,t^*)<D$, a contradiction. However, unlike the DAG algorithm, we cannot simply recurse independently on the graphs induced by $S_v$ and $T_v$ since the shortest path between a pair of vertices in $S_v$ may not be completely contained in $S_v$ (and analogously for $T_v$).

To overcome this hurdle, before recursing we first augment the graphs induced by $S_v$ and $T_v$ by carefully adding edges so that distances within these augmented graphs approximate distances in the original graph. 
%We add new weighted edges between $v$ and every vertex in $S_v$ and between $v$ and every vertex in $T_v$ and then we recurse on the graphs induced by $S_v\cup \{v\}$ and $T_v\cup \{v\}$. 
%The newly added weighted edges are defined as follows: 
Specifically, for every vertex $u\in S_v$, we add the directed edge $(u,v)$ with weight 0 and the directed edge $(v,u)$ with weight $\max\{0,d(v,u)-\eps(v)\}$. This choice of edges allows us to argue that 
%Then we argue that after adding these edges, 
%for any pair of vertices $u,w\in S_v$, $d_G(u,w)-2\eps(v)\leq d_{G_S}(u,w)\leq d_G(u,w)$. That is, 
the distances within the augmented graphs are approximations of the distances in $G$ up to an additive error of $2\eps(v)$. Then, by returning the maximum of $\eps(v)$ and the min-diameter estimates from recursing on the augmented graphs, we get an approximation guarantee, which turns out to be a logarithmic factor. Intuitively, the approximation factor is not constant because the recursion causes the distance distortion to compound at each level of recursion.

To reduce the approximation factor to a constant, we would like to decrease the number of recursion levels. To achieve this, we initially partition the graph into more than just two parts $S_v$ and $T_v$, by sampling more vertices. 
For our $\tO(m\sqrt{n})$ time 3-approximation, 
%instead of just doing full Dijkstra's algorithm from a single vertex $v$, 
we perform a full Dijkstra's algorithm from $\tO(\sqrt{n})$ vertices 
%and employ the recursive partition technique mentioned in the previous section 
to define an ordered partition of the vertices into $\tO(\sqrt{n})$ parts of $\tO(\sqrt{n})$ vertices each. 
%Our constant factor approximation algorithms are as follows. Our $\tO(m\sqrt{n})$ time 3-approximation algorithm runs Dijkstra's algorithm from $\tO(\sqrt{n})$ vertices (chosen using the recursive method outlined in the previous section) to define an ordered partition of the vertices into $\tO(\sqrt{n})$ parts of $\tO(\sqrt{n})$ vertices each. 
Then we apply the above idea of adding weighted edges within each part, however we must refine the definition of the graph augmentation to take into account \emph{all} of the $\tO(\sqrt{n})$ vertices we initially perform Dijkstra's algorithm from, instead of just $v$. Finally we use brute force (without recursion) on each part in the partition by running an exact all-pairs shortest paths algorithm. 

To achieve our time-accuracy trade-off algorithm, we carefully combine ideas from the logarithmic factor approximation and the 3-approximation algorithms. Specifically, we initially perform Dijkstra's algorithm from fewer than $\sqrt{n}$ vertices to define an ordered partition with larger parts than in the 3-approximation. Then we augment the graph induced by each part and carry out a constant number of recursion levels to further partition the graph before applying brute-force. 

\subsection{Min-Radius}
\subsubsection*{Algorithm for DAGs}
We begin by outlining the $\tO(m\sqrt{n})$ time 3-approximation algorithm for Min-Radius on DAGs from \cite{AbboudWW16}, which is very different from and more involved than the Min-Diameter algorithm on DAGs. 
%The algorithm uses the parameter $R$, the true min-radius. In reality, we do not know $R$ but we perform a binary search for $R$ by running the algorithm repeatedly. 
We begin by considering a topological ordering of the vertices and performing Dijkstra's algorithm from a set $W$ of $\tO(\sqrt{n})$ evenly spaced vertices including the first and last vertex. If a vertex $v\in W$ has min-eccentricity at most twice the true min-radius $R$ then we have obtained a 2-approximation. (We do not know $R$ in advance but we repeatedly run the algorithm with different values of $R$ to perform a binary search on $R$.)

Otherwise, we will define intervals in the ordering such that the min-center $c$ cannot be contained in any of these intervals. A key observation is that if there is a pair of vertices $(u,v)$ such that $u$ appears before $v$ in the topological ordering and $d(u,v)>2R$, then the min-center $c$ cannot fall between $u$ and $v$ in the topological ordering. This is because if it did, then $d(u,c)\leq R$ and $d(c,v)\leq R$, so $d(u,v)\leq 2R$, a contradiction.  We define the intervals that cannot contain $c$ as follows: for all $v\in W$ we let $a_v$ be the first vertex in the ordering such that $d(a_v,v)>2R$ (if it exists, otherwise $a_v=v$) and define $b_v$ to be the last vertex in the ordering such that $d(v,b_v)>2R$ (if it exists, otherwise $b_v=v$). Then, the key observation implies that $c$ cannot fall in the interval $[a_v,b_v]$ in the ordering. Now, we have a set of possibly overlapping intervals that cannot contain $c$. We take the union of these intervals to get a set of disjoint intervals that cannot contain $c$. 

%Let $\mathcal{I}$ be the set of intervals defined by the union of all vertices in some previously defined interval ($[a_v,v]$ or $[v,b_v]$ for some $v\in W$) i.e. we merge intersecting intervals. 
%Let $C$ be the set of potential min-centers i.e. the vertices that do not fall into any interval in $\mathcal{I}$. Each vertex $u\in C$ falls between two intervals $I_u,I'_u\in\mathcal{I}$. 
Every vertex $u$ that does not appear in such an interval, falls between two consecutive intervals $I_u$ and $I_u'$. We define the partial search graph of $u$ to be the graph induced by the set of vertices in $I_u$ or $I_u'$ or between $I_u$ and $I_u'$. 
%and between (introduced in Section \ref{sec:tech}) to be the graph induced by the vertices in its neighboring intervals in the interval the minimal interval of vertices that contains both $I_u$ and $I'_u$. 
After performing the partial searches, 
%we take the vertex with the smallest min-eccentricity with respect to its partial search graph, and we return 3 times this value. 
the algorithm returns 3 times the minimum min-radius of all partial search graph.
%minimum min-eccentricity $\min_{v}3\eps_{G_v}(u),\min_{v\in W}\eps_G(v)\}$. 
Next we give the idea of the analysis, which demystifies the factor of 3 in the returned value.

We claim that if the min-eccentricity of a vertex with respect to its partial search graph is at most $R$, then its min-eccentricity with respect to the full graph is at most $3R$,
%$\eps_{G_u}(u)\leq R$, then $\eps_G(u)\leq 3R$. 
and the min-eccentricity of the true min-center with respect to its partial search graph is at most $R$
%$\eps_{G_c}(c)\leq R$; 
(because for any path in a DAG whose starting and ending points are in a certain interval, every vertex in the path is in that interval). Thus, assuming the claim, $3R$ is a 3-approximation for the min-radius. We now outline the proof of the claim. Let $u$ be the min-center with the minimum min-radius $R$ of all partial search graphs. Let $v\in W$ such that $a_v$ is the first vertex (in the topological order) of $I_u$, then $v\in I_u$ and
%be a vertex in $W\cap I_u$.
%that appears before $u$ in the ordering. 
%By assumption, 
$d(v,u)\leq R$. Furthermore, by the definition of $a_v$, all vertices that appear before the beginning of the interval $I_u$ have distance at most $2R$ to $v$, and thus distance at most $3R$ to $u$. A symmetric argument holds for vertices that appear after the end of the interval $I_u'$. Hence the min-eccentricity of $u$ with respect to the full graph is at most $3R$.

This algorithm runs in time $O(m\sqrt{n})$ because the vertices of $W$ are evenly spaced so there are no more than $\sqrt{n}$ vertices between each pair of consecutive intervals. This implies that in the partial searches, each edge is only scanned $O(\sqrt{n})$ times. (Furthermore, repeatedly running the algorithm to binary search for $R$ adds a logarithmic factor to the runtime.)
    
\subsubsection*{Algorithm for general graphs}
We now give a high-level outline of our $\tO(m\sqrt{n})$ time 3-approximation algorithm for Min-Radius. This algorithm is much more delicate than our Min-Diameter algorithm, hence more of the details are deferred to the full description. We begin by running Dijkstra's algorithm from a set $W$ of $\tO(\sqrt{n})$ randomly sampled vertices to recursively partition the vertices into $S_v$ and $T_v$ as outlined in Section \ref{sec:tech}. This defines an initial DAG-like structure, however our analysis requires constructing a much more refined DAG-like structure. 
%Defining this structure takes a number of steps and it is not immediately clear how the first steps will eventually help to define this structure. 

Perhaps counter-intuitively, it makes sense to place vertices that are \emph{far} from each other in the graph \emph{close} to each other in the DAG-like structure. The reason for this is illuminated by the Min-Radius algorithm on DAGs, in which we find pairs of vertices $u,v$ that are far from each other and apply the key observation that the min-center cannot be between $u$ and $v$ in the topological ordering. Intuitively, it is as if we collapse the interval between $u$ and $v$ in the DAG since we do not have to search within this interval for the min-center. An analogous key observation is true for general graphs: if there is a pair of vertices $(u,v)$ with $d_{min}(u,v)>2R$, then either $c\in S_u\cap S_v$ or $c\in T_u\cap T_v$. This is because if $c\in T_u\cap S_v$, then $d(u,c)\leq R$ and $d(c,v)\leq R$ so $d(u,v)\leq 2R$, a contradiction; the last case $c\in S_u \cap T_v$ is symmetric. In our algorithm for general graphs, we ensure that far vertices are near each other in the DAG-like structure by doing the following: we let the \emph{far graph} $G_{far}$ be an undirected graph on $V$ with an edge between $u\in W$ and $v\in V$ if $d_{min}(u,v)>2R$. All vertices in $W$ that are in the same connected component in $G_{far}$ will be grouped in the DAG-like structure. We let $F_i$ be the set of vertices in $W$ that are in the $i^{th}$ connected component of $G_{far}$.

To construct the DAG-like structure, we show that precisely chosen groups of $F_i$s can be merged to create \emph{supercomponents}, which constitute a DAG-like structure in the following sense: there is an ordering of supercomponents such that for every pair of vertices $u,v\in W$ where the supercomponent containing $u$ appears before that containing $v$, $d(u,v)$ is small and $d(v,u)$ is large. Specifically, we define the \emph{close graph} $H$ whose vertex set is the set of $F_i$s. We add a directed edge between a pair of vertices in $H$ if there exists a short path (length $\leq 5R$) between the corresponding $F_i$s. Then we merge all $F_i$s that appear in the same strongly connected component of $H$ into a supercomponent. This contraction of strongly connected components of $H$ results in a DAG, which defines the ordering of the supercomponents.

Now that we have arranged the vertices in $W$ into a DAG-like structure, we would like to fit every vertex in the graph into this structure. Based on the precise way that we have defined the supercomponents, we can use an intricate argument to show \emph{roughly} the following property: for every vertex $v$ there exists an $i$ such that for every vertex $u\in W$ in the first $i$ supercomponents, $d(u,v)$ is small and for every vertex $u\in W$ in the remaining supercomponents, $d(v,u)$ is small.

After fitting every vertex into the refined DAG-like ordering, we can define each partial search graph to be an interval in the ordering that is large enough to contain several supercomponents.
%It is important that every vertex fits into the ordering because this allows us to define the partial search graphs $G_v$. Each $G_v$ is defined as an interval in this ordering that is large enough to contain several supercomponents.
%todo ref future lemmas from above?
%We claim that if the min-eccentricity of a vertex with respect to its partial search graph is at most $R$ then its min-eccentricity with respect to the full graph is at most $3R$.
%$\eps_{G_u}(u)\leq R$, then $\eps_G(u)\leq 3R$. 
%Furthermore, the min-eccentricity of the true min-center with respect to its partial search graph is at most $R$;
%(1) the min-eccentricity of the min-center with respect to its partial search graph is at most its min-eccentricity with respect to the original graph, and (2) if a vertex has small min-eccentricity within its partial search graph then it also has small min-eccentricity with respect to the original graph.
In the algorithm for DAGs, there were two important properties of the partial search graphs: (1) the min-eccentricity of the true min-center with respect to its partial search graph is at most $R$, and (2) if the min-eccentricity of a vertex with respect to its partial search graph is at most $R$ then its min-eccentricity with respect to the full graph is at most $3R$.
%$\eps_{G_c}(c)\leq R$, 
%and (ii) 
%if $\eps_{G_u}(u)\leq R$, then $\eps_G(u)\leq 3R$. 
We show that due to the precise structure of the supercomponents, refinements of properties (1) and (2) are also true for general graphs. 

Intuitively, property (1) is roughly true because for every pair of vertices $u,v\in W$ such that $u$'s supercomponent appears before $v$'s in the ordering, $d(v,u)> 5R$, since otherwise this pair of supercomponents would be in the same strongly connected component of $H$ and would have been merged into a single supercomponent. This implies that paths of length at most $R$ to or from the min-center cannot stray beyond its partial search graph. Intuitively, property (2) is roughly true because for every pair of vertices $u,v\in W$ such that $u$'s supercomponents appears before $v$'s in the ordering, $d(u,v)\leq 2R$ because otherwise, $u$ and $v$ would be in the same component of $G_{far}$ and thus be in the same supercomponent. Thus, like the argument for DAGs, for all $u$, all vertices that appear before $u$'s partial search graph $G_u$ have distance at most $2R$ to each supercomponent in $G_u$, and thus distance at most $3R$ to $u$. A symmetric argument holds for vertices after $u$ in the ordering. 

%For property (2), a modification of it is true for general graphs: if a vertex $u$ is within min-distance $R$ from all vertices in a carefully chosen \emph{subset} of $G_u$, with respect to distances in $G_u$, then the min-eccentricity of $u$ with respect to the full graph is at most $3R$. %$\eps_G(u)\leq 3R$. 
%Intuitively, this is true because for every pair of vertices $u,v\in W$ such that $u$'s supercomponents appears before $v$'s in the ordering, $d(v,u)> 5R$, since otherwise this pair of supercomponents would be in the same strongly connected component of $H$ and would have been merged into a single supercomponent. Thus, if a path of length at most $R$ contains both $u$ and $v$, then $u$ must appear before $v$ in the path. This property ensures that for any path of length at most $R$ whose starting and ending points are in a certain interval, the path does not stray too far outside of this interval. Thus, if each local search graph is defined by a large enough interval, every path of length at most $R$ from each vertex $v$ must be comtained in $G_v$. %[todo ref lemma in pf?]

%todo should it be G_{close}?

%mention above sections of algorithm: construction of the far graph, merging of far graph components into a DAG of super-components, placement of vertices into ordering, truncated BFS. Or just two with construction of DAG-like structure and local search. 
%todo in body of paper should include figures for both diameter and radius. 
%remember to state that we do binary search for R in radius alg.

\subsection{Min-Eccentricities}

Our Min-Eccentricities algorithm is a modification of our Min-Radius algorithm. In our Min-Radius algorithm, we identify a vertex whose min-eccentricity is at most about $3R$, where $R$ is the true min-radius. In our Min-Eccentricities algorithm, we show that with some extra bookkeeping, the algorithm can identify \emph{all} vertices with min-eccentricity at most about $5\rho$ for any $\rho$. We run the algorithm repeatedly, increasing $\rho$ by a factor of $(1+\delta)$ at each execution until we have estimated the min-eccentricity of every vertex.

The major modification of the Min-Radius algorithm here is that if one of the vertices that we run Dijkstra from has min-eccentricity at most $3\rho$, we cannot stop running the algorithm, as we can in the Min-Radius algorithm. Instead, we use this vertex as a tool to find vertices with min-eccentricity at most $5\rho$.

%A refinement of this step of the algorithm in the case of unweighted graphs gives us a $(3+\rho)$-approximation algorithm.%and then we move to the next step of the algorithm to further detect vertices with low min-eccentricity. 

%todo make sure there are no epsilons in this section or the intro (or define it)

\section{Preliminary Graph Partitioning}\label{section:prelim}
%todo rename this section?

In this section we describe a graph partitioning procedure we use as a first step in our Min-Diameter, Min-Radius, and Min-Eccentricities algorithms. The goal of this partitioning is to define a DAG-like structure in general directed graphs.

%In all of our algorithms we begin by partitioning the vertices with the goal of defining a DAG-like structure in the graph. In this section we formally define the partitioning procedure.

\begin{definition}
Assign each vertex a unique ID from $[n]$. For each vertex $v$, let $S_v=\{u\in V: d(u,v)<d(v,u)\lor [d(u,v)=d(v,u)\land ID(u)<ID(v)]\}$. Let $T_v=V\setminus (S_v\cup\{v\})$. 
\end{definition}

%Let $d(u,v)$ be the length of the shortest path from $u$ to $v$, and $d_{min}(u,v)=\min\{d(u,v),d(v,u)\}$. 

%For any distance $r>0$: let $S_v^r=\{u\in S_v: d(u,v)\le r\}$, $T_v^r=\{u\in T_v: d(v,u)\le r\}$. For a set $W$ of vertices, let $S_W=\bigcap_{v\in W}S_v$, $T_W=\bigcap_{v\in W}T_v$, $S_W^r=\bigcap_{v\in W}S_v^r$, and] $T_W^r=\bigcap_{v\in W}T_v^r$.

%Next we motivate the formal definition of $S_v,T_v$. Recall that their intersections are essentially the potential subgraph $G_v$'s, which need to be small. One approach is to always have some $v$ such that $S_v$ and $T_v$ are of similar size, and by induction so are the final intersections. As mentioned in the introduction, to obtain a partial order, we have $M_{u,v}=0$ if $d(u,v)<d(v,u)$ and $M_{u,v}=1$ if $d(u,v)>d(v,u)$, and $M_{u,v}$ where $d(u,v)=d(v,u)$ remains to be defined. Suppose we simply define $S_v=\{u:d(u,v)\le d(v,u)\}$, then $M$ is no longer antisymmetric ($M_{u,v}=1-M_{v,u}$ for all $u,v$), and for complete graphs $T_v$'s are empty. In fact, the following lemma shows that as long as $M$ is antisymmetric, there is a constant fraction of vertices $v$ such that $S_v$ and $T_v$ are of similar size:

The runtime of our algorithms relies on whether the partition into $S_v$ and $T_v$ is \emph{balanced}. Using the observation that if $u\in S_v$, then $v\in T_u$, the following lemma shows that for most vertices, the partition is indeed approximately balanced. 

\begin{lemma}
\label{lem:antisym}
%Suppose $M$ is antisymmetric. For any constant $c>1$, the fraction of vertices $v$ such that $\frac{\abs{S_v}}{\abs{T_v}}in (\frac 1c,c)$ is at least $\frac {c-3}{c+1}$; specifically, 
For any graph on $n$ vertices there are more than $\frac{n}{2}$  
%$\frac n4$ 
vertices $v$ such that $\frac{\abs{S_v}}{8}\le\abs{T_v}\le8\abs{S_v}$. %$\frac{\abs{S_v}}{\abs{T_v}}\in(\frac 18,8)$.

More generally, for any $U\subseteq V$, there are more than $\frac{\abs{U}}{2}$  
%$\frac n4$ 
vertices $v\in U$ such that $\frac{\abs{S_v\cap U}}{8}\le\abs{T_v\cap U}\le8\abs{S_v\cap U}$.
\end{lemma}

\begin{proof}
Since the first statement is a special case of the second statement with $U=V$, we prove the more general statement. Let $\abs{U}=k$.
Let $M$ be a $k\times k$ matrix indexed by the vertices in $U$ where $M_{u,v}=-1$ if $u\in S_v\cap U$, $M_{u,v}=1$ if $u\in T_v\cap U$, and $M_{u,u}=0$ for $u\in U$. Note that $M$ is skew-symmetric, i.e., $M_{u,v}=-M_{v,u}$ for all $u,v$. For any $A,B\subseteq U$, let $M_B$ be the $k\times\abs{B}$ submatrix consisting of the columns indexed by $B$, and let $M_{A,B}$ the $\abs{A}\times\abs{B}$ submatrix of $M_B$ consisting of its rows indexed by $A$. 

Suppose for contradiction there is a set $C\subset U$ of $\frac{k}{4}$ vertices $v$ such that $\abs{T_v\cap U}>8\abs{S_v\cap U}$. Then $M_C$ contains at least $\frac 89 k\cdot\frac k4=\frac 29 k^2$ ones.

The $\frac k4\times\frac k4$ submatrix $M_{C,C}$ is also skew-symmetric, so at most half of its entries are ones, i.e., $M_{C,C}$ contains at most $\frac{k^2}{32}$ ones. Letting $\bar{C}=U\setminus C$, we see that $M_{\bar{C},C}$ has $\frac{3}{4}k \times \frac{k}{4} = \frac{3}{16} k^2$ entries, and hence at most $\frac{3}{16} k^2$ ones. In total, $M_C$ contains at most $\frac{7}{32} k^2<\frac{2}{9} k^2$ ones, contradiction. 

Therefore the number of vertices $v\in U$ such that $\abs{T_v\cap U}>8\abs{S_v\cap U}$ is less than $\frac{k}{4}$, and symmetrically the number of vertices $v\in U$ such that $\abs{T_v\cap U}<\frac{\abs{S_v\cap U}}{8}$ is less than $\frac{k}{4}$. Hence more than half of the vertices $v\in U$ have that $\frac{\abs{S_v\cap U}}{8}<\abs{T_v\cap U}<8\abs{S_v\cap U}$.
\end{proof}

%Recursively applying this balanced partition, the next lemma shows that $V$ can be partitioned into small subsets that have a DAG-like structure.

Next, we describe how we use Lemma~\ref{lem:antisym}  to recursively construct a balanced partition of the vertices into a given number of of sets.

\begin{lemma}
\label{lem:sampling}
Given a graph $G$ with $n$ vertices and a constant $c>0$, in $\tO(mn^{1-c})$ time we can partition $V$ into disjoint sets $W,V_1$, $V_2$,\ldots,$V_{q+1}$, where $q=\abs{W}=n^{1-c}$, such that with high probability:

%can find $q=n^{1-c}$ vertices $v_1, v_2,\ldots, v_q$ and partition $V\setminus\bigcap_i\{v_i\}$ into disjoint subsets $V_1$, $V_2$,…,$V_{q+1}$ such that with high probability, 
\begin{enumerate}
    \item for all $i$, $\abs{V_i}=\Theta(\frac nq)$;
    \item for all $i\ne j$, there exists a vertex $w\in W$ such that either $V_i \subseteq S_{w}$,$V_j \subseteq T_{w}$, or $V_i \subseteq T_{w}$,$V_j \subseteq S_{w}$;
    \item for all $U\subseteq W$, let $V_U=\displaystyle{\left(\bigcap_{w\in U}S_w\right)\bigcap\left(\bigcap_{w\in W\setminus U}T_w\right)}$, then $V_U\subseteq V_i$ for some $i\in[q+1]$.
    %\item for all $U\subseteq W$, $\displaystyle{\left(\bigcap_{w\in U}S_w\right)\bigcap\left(\bigcap_{w\in W\setminus U}T_w\right)}\subseteq V_i$ for some $i\in[q+1]$.
    %todo make all sample stuff W
    
%that $V_i \subseteq S_{v_k}$ and $V_j \subseteq T_{v_k}$, or $V_i \subseteq T_{v_k}$ and $V_j \subseteq S_{v_k}$
   % \item for any $Z \subseteq V$ such that for all $v_i$ either $Z \subseteq S_{v_i}$ or $Z \subseteq T_{v_i}$ then $\abs{Z}=O(\frac nq)$.
\end{enumerate}
%todo move prop 3 to min-radius analysis

%for all $G_i$ and $v_j$, all vertices in $G_i$ are either from $S_{v_j}$ or from $T_{v_j}$.

\end{lemma}
\begin{proof}
We begin with $W=\emptyset$ and we will iteratively populate $W$ with vertices. We let $\mathcal{V}_0=\{V\}$ and for all $i\in [q]$ when we add the $i^{th}$ vertex to $W$, we will construct $\mathcal{V}_{i}$ from $\mathcal{V}_{i-1}$ by partitioning the largest set in $\mathcal{V}_{i-1}$ into two parts.  After adding $q$ vertices to $W$ we will have constructed $\mathcal{V}_{q}=\{V_1\dots V_{q+1}\}$. 

%We add our first vertex to $W$ and perform our partition of $\mathcal{V}_0$ as follows. By Lemma \ref{lem:antisym}, if we randomly sample $O(\log^2 n)$ vertices, with probability $1-2^{-\log^2 n}=1-n^{-\log n}$ we will sample a vertex $w_1$ such that $S_{w_1}$ and $T_{w_1}$ differ by a factor of at most $8$. We add $w_1$ to $W$ and we let $\mathcal{V}_1=\{S_{w_1},T_{w_1}\}$.
%$$W=\{w_1\}$ such that $V_1=S_{w_1}$ and $V_2=T_{w_1}$ differ by a factor of at most $8$. 

For all $i\in [q]$, let $A_i,B_i$ be the largest and smallest sets in $\mathcal{V}_i$, respectively.

We describe how to construct $W$ and $\mathcal{V}_q$ inductively. Suppose $\abs{W}=r-1$ and we have constructed $\mathcal{V}_{r-1}$. 
By Lemma~\ref{lem:antisym}, if we randomly sample $O(\log^2 n)$ vertices from $A_{r-1}$, with probability at least $1-2^{-\log^2 n}=1-n^{-\log n}$ we will sample a vertex $w_r$ such that $A_S=A_{r-1}\cap S_{w_r}$ and $A_T=A_{r-1}\cap T_{w_r}$ differ by a factor of at most $8$. We add $w_r$ to $W$ and let $\mathcal{V}_r=\mathcal{V}_{r-1}\cup\{A_S,A_T\}\setminus\{A_{r-1}\}$. %Hence after splitting with $w_q$, $\abs{W}=q$ and $\abs{\mathcal{V}_q}=q+1$.

By union bound over the $q=n^{1-c}$ partitionings, with probability at least $1-n^{1-c-\log n}$, every partitioning produces two sets that differ in size by a factor of at most 8.
%todo change statement of lemma to say actual probability and delete whp from here.

%For $r\in[q]$, suppose inductively that we have partitioned $V$ into $W,V_1,\dots,V_{r-1}$,let $V_a$ be the largest $V_i$ for $i\in[r-1]$. Applying the analysis in Lemma \ref{lem:antisym} to $M_{V_a}$, we can find $w_r\in V_a$ and split $V_a$ into $V_a\cap S_{w_r}$ and $V_a\cap T_{w_r}$, whose sizes differ by a factor of at most $8$ with probability $1-n^{-\log n}$, and add $w_r$ to $W$. By union bound over the $q=n^{1-c}$ rounds, all splits are balanced (the two resulting subsets differ by a factor of at most 8) with probability $1-n^{1-c-\log n}$, i.e, with high probability.
%Any other old subset $V_i^{r-1}$ is also split into $V_i^{r-1}\cap S_{v_r}$ and $V_i^{r-1}\cap T_{v_r}$.

We prove property 1 by induction on $\abs{W}=r$. Specifically, we will show that for all $r\in[q]$, $|A_r|\leq 9|B_r|$. This implies that $\abs{A_q}=O(\abs{B_q})$, and property 1 follows.
%=O\left(\frac{q\abs{B_q}}{q}\right)=O\left(\frac{\sum\abs{V_i}}{q}\right)=O(\frac nq)$. 
Lemma \ref{lem:antisym} implies that $|A_1|\leq 9|B_1|$. Assume inductively that $\abs{A_{r-1}}\le 9\abs{B_{r-1}}$. Since no subset grows in size, $\abs{A_r}\le\abs{A_{r-1}}$ and $\abs{B_r}\le\abs{B_{r-1}}$. If $\abs{B_r}=\abs{B_{r-1}}$, then $\abs{A_r}\le \abs{A_{r-1}}\le9\abs{B_{r-1}}=9\abs{B_r}$. Otherwise, $\abs{B_r}<\abs{B_{r-1}}$, which implies that $B_r$ is one of the two sets obtained by partitioning $A_{r-1}$. In this case $\abs{A_{r-1}}\le 9\abs{B_r}$ by Lemma \ref{lem:antisym}. Hence $\abs{A_r}\le\abs{A_{r-1}}\le9\abs{B_r}$, completing the induction.

Property 2 follows from the partitioning procedure: for any $i\ne j$, if for all $w\in W$, $V_i,V_j\subseteq S_{w}$ or $V_i,V_j\subseteq T_{w}$ then $V_i\cup V_j$ would never have been partitioned.

Property 3 also follows from the partitioning procedure: observe that for all $w\in W$ and all $U\subseteq W$, $V_U\subseteq S_w$ or $V_U\subseteq T_w$, so $V_U$ is never partitioned and thus $V_U\subseteq V_i$ for some $i\in[q+1]$.

%any such $V_U$ is a subset of one of the $V_i$'s and hence satisfies $\abs{Z} \le \abs{V_i} = (\frac nq)$.

Since we sample $n^{1-c}\log^2 n$ vertices and for all $v$ finding $S_v,T_v$ takes $O(m)$ time, the runtime is $\tO(mn^{1-c})$. 

\end{proof}

\section{Min-Diameter Algorithm}\label{section:diam}
Throughout this section, let $D$ be the min-diameter, and let $s^*,t^*$ the endpoints of the min-diameter. In this section we prove the time/accuracy trade-off theorem for Min-Diameter.
\begin{theorem}\label{thm:min_diam_param}
For any integer $0<\ell\le O(\log{n})$, there is an $\tilde{O}(mn^{1/(\ell+1)})$ time randomized algorithm that, given a directed weighted graph $G$ with edge weights non-negative and polynomial in $n$, can output an estimate $\tilde{D}$ such that $D / (4\ell-1) \le \tilde{D} \le D$ with high probability, where $D$ is the min-diameter of $G$. 
\end{theorem}

%Before describing the algorithm, we prove some useful lemmas, and then we offer a special case of the algorithm when $\ell=1$ in Theorem~\ref{thm:min_diam_no_param}. Fianlly, we will use Theorem~\ref{thm:min_diam_no_param} to prove Theorem~\ref{thm:min_diam_param}.

We first prove a special case of Theorem~\ref{thm:min_diam_param} where $\ell=1$.

\subsection{An $\tO(m\sqrt{n})$ time 3-approximation}

\begin{theorem}{{\bf (Theorem~\ref{thm:min_diam_param}} with $\ell=1$)} %\label{thm:min_diam_no_param}
There is an $\tilde{O}(m\sqrt{n})$ time randomized algorithm, that given a directed weighted graph $G=(V,E)$ with edge weights non-negative and polynomial in $n$, can output an estimate $\tilde{D}$ such that $D / 3 \le \tilde{D} \le D$ with high probability, where $D$ is the min-diameter of $G$.\label{thm:min_diam_no_param} 
\end{theorem}

\subsubsection{Algorithm Description}

%todo make notation consistent with lem 3.2, also for min-rad
Applying Lemma~\ref{lem:sampling} with $q=\sqrt{n}$ we obtain a partition of the vertices into $W, V_1, V_2, \ldots, V_{\sqrt{n}+1}$.%, and $\sqrt{n}$ vertices $w_1, w_2, \ldots, w_{\sqrt{n}}$.
%Combined with Lemma~\ref{lem:splitting_v}, this means if $s^*$ and $t^*$ are in $V_i$ and $V_j$ where $i\ne j$, then $D \ge \eps(u_k) \ge D/2$. 

We perform Dijkstra's algorithm from every vertex in $W$ and define $D' =\max_{w\in W}\eps(w)$.
%\max_{1 \le k \le \sqrt{n}} \eps(w_k)$. 
We will later show that $D'$ is a good approximation of the Min-Diameter when $s^*$ and $t^*$ are not in the same vertex set $V_i$. 

For every $i \in [\sqrt{n}+1]$, define $W_i^S = \{w\in W: V_i \subseteq S_{w}\}$, and $W_i^T = \{w\in W: V_i \subseteq T_{w}\}$. Then, for every $i$, we construct two graphs $G_i^S$ and $G_i^T$. The first graph $G_i^S$ contains all vertices of $V_i$ and an additional node $w_i^S$. It has the following edges:
\begin{enumerate}
    \item For every directed edge $(u, v)\in E$ such that $u, v \in V_i$, add this edge to $G_i^S$.
    \item Add a directed edge from $w_i^S$ to every $v \in V_i$, with weight $\max\left\{\min_{w \in W_i^S}d(w, v) - D', 0\right\}$, and a directed edge from every $v \in V_i$ to $w_i^S$ with weight $0$.
\end{enumerate}
The second graph $G_i^T$ is symmetric to $G_i^S$. It contains all vertices in $V_i$ and an additional node $w_i^T$. It has the following edges:
\begin{enumerate}
    \item For every directed edge $(u, v) \in E$ such that $u, v \in V_i$, add this edge to $G_i^T$.
    \item Add a directed edge from every $v \in V_i$ to $w_i^T$, with weight $\max\left\{\min_{w \in W_i^T} d(v, w) - D', 0\right\}$, and add a directed edge from $w_i^T$ to every $v \in V_i$ with weight $0$. 
\end{enumerate}

%For every vertex set $V_i$, we construct the graph $G_i^S$ and $G_i^T$ as described above, and computes the distances $d_{G_i^S}(a,b)$ and $d_{G_i^T}(a,b)$ for every pair of vertices $a,b\in V_i$. 
For all $i$, we run an exact all-pairs shortest paths algorithm on $G_i^S$ and $G_i^T$. This allows us to compute for all $i$ and all $u,v\in V_i$ the quantity $\min\{d_{G_i^S}(u,v), d_{G_i^T}(u,v)\}$, which we denote by $d'_i(u,v)$.

We choose the larger between $D'$ and $\max_{i\in[\sqrt{n}+1],u,v \in V_i} \min\{d'_i(u, v),d'_i(v,u)\}$ as our final estimate for the min-diameter.

\subsubsection{Analysis}
%In this section, we describe the Min-Diameter algorithm in full detail. We use the balanced sampling technique described in Lemma \ref{lem:sampling}. Based on the size of the vertices we sample at each level of recursion, we get a tradeoff between runtime and approximation ratio of the algorithm. Specifically, we get the following theorem.
The following lemma will be used to show that $D'$ is a good estimate for the min-diameter if $s^*$ and $t^*$ happen to fall into different sets $V_i$
%happen to fall into different sets $S_{w_k}$, $T_{w_k}$ for some $k$.
%todo is the above sentence helpful or confusing?

\begin{lemma}\label{lem:splitting_v}
For all vertices $v$, if either $s^* \in S_v$, $t^* \in T_v$, or $t^* \in S_v$, $s^* \in T_v$, then $\eps(v) \ge D/2$. 
\end{lemma}
\begin{proof}
We only consider the case when $s^* \in S_v$ and $t^* \in T_v$ as the other case is symmetric. By way of contradiction, assume that $\eps(v) < D/2$, then we have $d_{\min}(s^*, v) < D/2$ and $d_{\min}(t^*, v) < D/2$. Since $s^* \in S_v$, $d(s^*, v) = d_{\min}(s^*, v) < D / 2$; similarly, since $t^* \in T_v$, $d(v, t^*) = d_{min}(t^*, v) < D / 2$. Therefore, by the triangle inequality, $d(s^*,t^*)<D$, a contradiction.
%if we start from $s^*$, go through $v$, and then reach $t^*$, we have a path of length less than $D$. This means the min distance between $s^*$ and $t^*$ is less than $D$, which is a contradiction. 
\end{proof}

%The previous lemma is used for the case where $s^*$ and $t^*$ fall into different sets $S_{w}$, $T_{w}$ for some $w\in W$, while the next 

The next two lemmas are used for the case where $s^*$ and $t^*$ fall into the same set $V_i$.
%todo is it two?

\begin{lemma}\label{lem:diam_upper_bound}
For every $i$, and every pair of vertices $u, v \in V_i$, $d'_i(u,v)\le d(u, v)$; that is, \\$\min\{d_{G_i^S}(u,v), d_{G_i^T}(u,v)\}\leq d(u,v)$.
%\min\{ d_{G_i^S}(a,b), d_{G_i^T}(a,b)\}$. 
\end{lemma}
\begin{proof}
Take any shortest path in the original graph $G$ from $u$ to $v$. If this path does not leave $V_i$, then this path also exists in $G_i^S$ and $G_i^T$, and thus the inequality is true. 

 \begin{figure}
   \centering
     \includegraphics[width=0.3\textwidth]{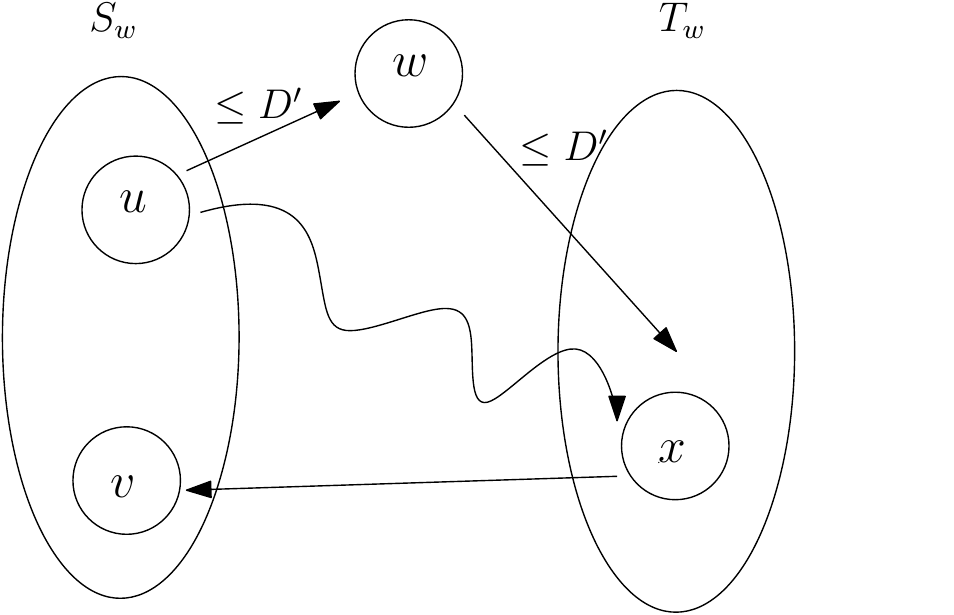}
     \caption{The case where $u,v\in S_w$ and the shortest path from $u$ to $v$ contains a node $x\in T_w\cup \{w\}$.}
     %, which is defined in more detail in Lemma \ref{lem:diam_upper_bound}.}
 	\label{fig:diam_upper_bound}
 \end{figure}
 
It remains to prove for the case when the shortest $u,v$ path in the original graph leaves $V_i$. Let $x\not\in V_i$ be any vertex on a shortest $u,v$ path.
%that is on a shortest $u,v$ path Assume the last vertex in this path outside of $V_i$ is $x$. 
By Lemma \ref{lem:sampling}, property 2, there exists $w\in W$ such that $x \in S_{w} \cup \{w\}$ and $V_i \subseteq T_{w}$, or $x \in T_{w} \cup \{w\}$ and $V_i \subseteq S_{w}$. We first assume $x \in T_{w} \cup \{w\}$ and $V_i \subseteq S_{w}$ as shown in Figure~\ref{fig:diam_upper_bound}, and the other case is symmetric. 
%todo update notation in fig

Since $x$ is on the shortest path from $u$ to $v$, we have $d(u,v)\ge d(x,v)$. Also, we have $d(w, x) \leq D'$, by definition of $D'$. Therefore, 
\begin{equation}\label{eq:eq_lemma1}
    \begin{split}
        d(u,v)&\ge d(x,v)\\
              &\ge d(x,v)+\left(d(w, x) - D'\right)\\
              %& = \left(d(w, x)+d(x,v)\right) -D'\\
              & \ge d(w, v)-D'
    \end{split}
\end{equation}
Now consider the path $u \rightarrow w_i^S \rightarrow v$ in $G_i^S$. The first part $u \rightarrow w_i^S$ costs $0$, because there is an edge from $u$ to $w_i^S$ with weight $0$; the second part $w_i^S \rightarrow v$ costs at most $\max\{0, d(w, v)-D'\}$. If $d(w, v) < D'$, then $d'_i(u,v)\le d_{G_i^S}(u,v) = 0 \le d(u,v)$; otherwise, $d'_i(u,v)\le d_{G_i^S}(u,v)  \le d(w, v)-D' \le d(u, v)$, where the last step is Equation \ref{eq:eq_lemma1}.

When $x \in S_{w} \cup \{w\}$, and $V_i \subseteq T_{w}$, we have a symmetric argument: $d(u,v)\ge d(u,x) \ge d(u, x) + \left(d(x, w)-D'\right) \ge d(u, w)-D'$. Consider the path $u\rightarrow w_i^T \rightarrow v$ in $G_i^T$. The second part $w_i^T \rightarrow v$ costs $0$, because there is an edge from $w_i^T$ to $v$ with weight $0$; the first part $u \rightarrow w_i^T$ costs at most $\max\{0, d(u, w)-D'\}$. If $d(u, w) < D'$, then $d'_i(u,v) \le d_{G_i^T}(u,v) = 0 \le d(u,v)$; otherwise, $d'_i(u, v)\le d_{G_i^T}(u,v)  \le d(u, w)-D' \le d(u, v)$. 
\end{proof}

\begin{lemma}\label{lem:diam_lower_bound}
For every $i$, and every pair of vertices $u, v \in V_i$, $d'_i(u,v)\geq d(u,v)-2D'$; that is, $d_{G_i^S}(u,v) \ge d(u,v)-2D'$ and $d_{G_i^T}(u,v) \ge d(u,v)-2D'$. 
\end{lemma}
\begin{proof}

We only provide full proof for $d_{G_i^S}(u,v) \ge d(u,v)-2D'$. The inequality for $G_i^T$ can be proved by a symmetrical argument.
If the shortest path from $u$ to $v$ in $G_i^S$ does not contain $w_i^S$, then this path also exists in the original graph $G$, and thus the inequality is true. 
 
  \begin{figure}
   \centering
     \includegraphics[width=0.3\textwidth]{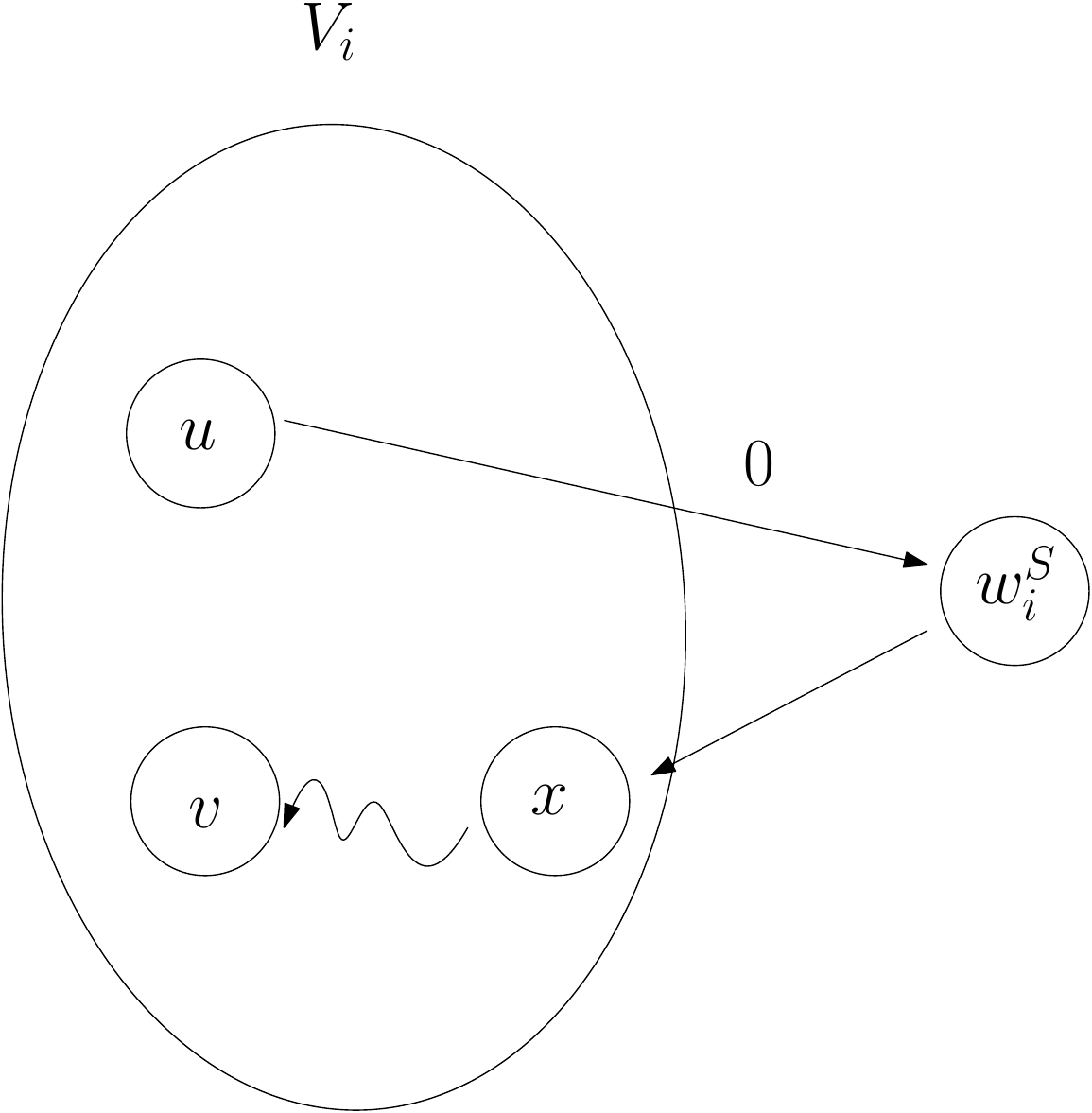}
     \caption{A shortest $u,v$ path in $G_i^S$ that contains $w_i^S$. The path goes from $u$, directly to $w_i^S$ using a weight 0 edge, then directly to a vertex $x$, and finally reaches $v$. }
 	\label{fig:diam_lower_bound}
 \end{figure}
 
Otherwise, the shortest path from $u$ to $v$ in $G_i^S$ contains $w_i^S$, as shown in 
%, so $d_{G_i^S}(u,v)$ can be split into two parts: $u\rightarrow w_i^S$ and $w_i^S \rightarrow v$, as shown in
Figure~\ref{fig:diam_lower_bound}. %By construction, $d_{G_i^S}(u,w_i^S)=0$. 
%The first part has $0$ cost because there is a zero-weighted edge from $u$ to $w_i^S$. For the second part, 
All edges on the shortest path from $w_i^S$ to $v$ exist in the original graph $G$ except for the first edge from $w_i^S$ to some node $x$, since a shortest path cannot use the vertex $w_i^S$ more than once. That is, $d_{G_i^S}(x,v)=d(x,v)$.  

By the definition of $w_i^S$ and the edges incident to it, there exists a $w \in W_i^S$ such that 
$d(w, x)\leq d_{G_i^S}(w_i^S, x) + D'$. Thus, we have
%Thus, 
%$d_{G_i^S}(w_i^S,x) \ge d(w, x)-D'$. Thus, $d_{G_i^S}(w_i^S, v) \ge d_{G_i^S}(w_i^S,x)+d(x,v)\geq d(w, x)-D'+d(x,v) \geq d(w, v)-D'$. By the triangle inequality, we have $d_{G_i^S}(u,v) \ge d_{G_i^S}(u,w_i^S)+d_{G_i^S}(w_i^S, v)\ge d(w, v)-D'$.

%Also, by definition of $D'$, $d(u, w) - D' \le 0$. Therefore, 
%todo delete all mentions of local search in the entire paper
%\begin{equation}
   % \begin{split}
   \begin{alignat*}{3}
        d_{G_i^S}(u,v) & = d_{G_i^S}(u,w_i^S)+d_{G_i^S}(w_i^S, x)+d_{G_i^S}(x, v) \\%&&\text{ decompose the path}\\
        &=d_{G_i^S}(w_i^S, x)+d_{G_i^S}(x, v) &&\text{ since $d_{G_i^S}(u,w_i^S)=0$ by construction}\\
        &=d_{G_i^S}(w_i^S,x)+d(x,v) &&\text{ from argument above}\\
        &\ge d(w, x)-D'+d(x,v) &&\text{ by the definition of $w$}\\
        &\ge d(w, v)-D' &&\text{ by the triangle inequality}\\
        & \ge \left(d(w, v)-D'\right) + \left(d(u, w) - D'\right) &&\text{ since $d(u,w)\leq D'$ by definition}\\
        %& = \left(d(u,w)+d(w,v)\right)-2D'\\
        &\ge d(u,v)-2D' &&\text{ by the triangle inequality}
        \end{alignat*}
   % \end{split}
%\end{equation}
\end{proof}

%We are now ready to prove the special case of theorem \ref{thm:min_diam_param} where $\ell=1$. Then we will use this to prove the result for general $\ell$. 

%Suppose $D$ is the true Min-Diameter. 
We are now ready to prove our approximation ratio guarantee: $D/3\leq \tilde{D}\leq D$.
Clearly $D'\le D$ because $D'$ is the min-eccentricity of a vertex. 
%$\max_{i, u \in V_i, v \in V_i} d'_i(u,v) \le D$ because By
By Lemma~\ref{lem:diam_upper_bound} $\max_{i, u \in V_i, v \in V_i} \min\{d'_i(u,v),d'_i(v,u)\} \le \max_{i, u \in V_i, v \in V_i} d_{min}(u,v) \le D$ . Therefore, we never over estimate the Min-Diameter. 

If $s^*\in W$ or $t^*\in W$, then since we run Dijkstra from all vertices in $W$ we have $D'=D$. So assuming that $s^*,t^*\notin W$, we have two cases.

\noindent {\bf Case 1:}
%If $D' \ge D/3$, then $D'$ is already a $3$-approximation, and we are done. Otherwise, assume $D' < D/3$. 
$s^*$ and $t^*$ are not in the same vertex set $V_i$. By Lemma~\ref{lem:sampling}, property 2, there exists $w \in W$ such that one of $s^*$ and $t^*$ is in $S_{w}$ and the other is in $T_{w}$, so by Lemma~\ref{lem:splitting_v}, $\eps(w) \ge D/2$. Since $D' \ge \eps(w)$, we have $D' \ge D/2$. 

\noindent{\bf Case 2:} $s^*$ and $t^*$ are in the same vertex set $V_i$ for some $i$.
 By Lemma~\ref{lem:diam_lower_bound}, $\min\left(d'_i(s^*,t^*),d'_i(t^*,s^*)\right)\ge d_{min}(s^*,t^*)-2D' = D-2D'$. Since $\max\{D-2D',D'\}\ge D/3$, we get a $3$-approximation.

\paragraph*{Runtime analysis}
%Next we perform the runtime analysis for different steps:
 It takes $\tilde{O}(m\sqrt{n})$ time to perform the partitioning from Lemma~\ref{lem:sampling} and to perform Dijkstra's algorithm from all $w \in W$ since $|W| = O(\sqrt{n})$.
    
    For all $i$, the number of vertices in $G_i^S$ is $\abs{V_i}+1=O(\sqrt{n})$ with high probability by property 1 of Lemma~\ref{lem:sampling} and the number of edges is $m_i+O(\sqrt{n})$ where $m_i$ is the number of edges in the graph induced by $V_i$. Hence we can run an all-pairs shortest paths algorithm on $G_i^S$ in time $\tilde{O}((m_i+\sqrt{n})\sqrt{n})$. Summing over all $i$ gives us $\tilde{O}(m\sqrt{n})$. The same analysis also works for $G^i_T$.

%So overall, the complexity is $\tilde{O}(m\sqrt{n})$.

%\end{proof}

\subsection{Time/accuracy trade-off algorithm}

\subsubsection{Algorithm Description}

We begin by briefly outlining the differences between our trade-off algorithm and our $O(m\sqrt{n})$ time algorithm. For our trade-off algorithm, instead of applying Lemma \ref{lem:sampling} to sample $q=\sqrt{n}$ vertices, we will apply Lemma \ref{lem:sampling} with a smaller value of $q$ to save time. This results in a smaller set $W$ and larger sets $V_i$. In our $O(m\sqrt{n})$ time algorithm, we had time to apply brute force (i.e. run all-pairs shortest paths) on the graphs $G_i^S$ and $G_i^T$, however in our trade-off algorithm we do not. Instead, we apply recursion. Simply constructing $G_i^S$ and $G_i^T$ and recursing on both of them does not suffice because each recursive call only returns the min-diameter, whereas we require knowing all distances. To overcome this issue, instead of constructing $G_i^S$ and $G_i^T$ separately, we construct a graph $G_i$ that combines these two graphs. Then, we show that it suffices to recurse on $G_i$ to compute only its min-diameter rather than all distances.

The algorithm is as follows. We apply Lemma \ref{lem:sampling} with $q=O(n^{1/(\ell+1)})$ to partition the vertices into $W, V_1, V_2, \ldots, V_{q+1}$. We perform Dijkstra's algorithm from every vertex in $W$ and define $D' =\max_{w\in W}\eps(w)$. For every $i \in [\sqrt{n}+1]$, we define $W_i^S = \{w\in W: V_i \subseteq S_{w}\}$, and $W_i^T = \{w\in W: V_i \subseteq T_{w}\}$. For every $i\in [q+1]$, we construct the graph $G_i$ as follows.
%for all pairs of vertices $u$ and $ v$, we need to recurse on smaller graphs and get estimated Min-Diameter of smaller graphs. However, since we are using the maximum of $\min\left( d_{G_i^S}(u, v), d_{G_i^T}(u, v) \right)$ to update the answer, it is tricky to decide which graph to recurse on. Therefore, in order to prove Theorem \ref{thm:min_diam_param}, we need to create $G_i$, which combines $G_i^S$ and $G_i^T$. 
The vertex set of $G_i$ is all vertices $V_i$ and two additional vertices $w_i^S$ and $w_i^T$. It contains the following edges:
\begin{enumerate}
    \item For every directed edge $(u, v)\in E$ such that $u, v \in V_i$, add this edge to $G_i$.
    \item Add a directed edge from $w_i^S$ to every $v \in V_i$, with weight $\max\{\min_{w \in W_i^S}d(w, v) - D', 0\}$, and add a directed edge from every $v$ to $w_i^S$ with weight $0$.
    \item Add a directed edge from every $v \in V_i$ to $w_i^T$, with weight $\max\{\min_{w \in W_i^T} d(v, w) - D', 0\}$, and add a directed edge from $w_i^T$ to every $v \in V_i$ with weight $0$. 
\end{enumerate}

For all $i$, we recursively compute a $(4\ell-5)$-approximation for the Min-Diameter of $G_i$ by calling the algorithm for $\ell-1$. We use the $\ell=1$ algorithm from the previous section as the base case. 
%One special case is when $\ell - 1 > O(\log |G_i|)$, so we cannot directly apply the theorem recursively because $\ell$ is out of range. However, if we decrease $\ell$ to $\log |G_i|$ in this case, we get a better approximation ratio while keeping the same runtime, since $\tilde{mn^{1/(\ell+1)}}=\tilde{m}$ for any $\ell \ge \log n$. 

We choose the larger between $D'$ and the maximum approximated Min-Diameter over all $G_i$ as our final estimate.

%the following two as the final estimate for Min-Diameter:
%\begin{itemize}
 %   \item Maximum min-eccentricity among all $w_i$. This is the value $D'$ defined above. 
  %  \item The maximum approximated Min-Diameter over all $G_i$. 
%\end{itemize}

\subsubsection{Analysis}
Before proving the main theorem for Min-Diameter,  we need to prove two lemmas for $G_i$, which are analogous to Lemma \ref{lem:diam_upper_bound} and Lemma \ref{lem:diam_lower_bound}.

\begin{lemma}\label{lem:diam_upper_bound_combined}
For every $i$, and every pair of vertices $u, v \in V_i$, $d(u, v) \ge d_{G_i}(u,v)$. 
\end{lemma}
\begin{proof}
Since $G_i^S \subseteq G_i$ and $G_i^T \subseteq G_i$, we have $d_{G_i}(u,v) \le d_{G_i^S}(u,v)$ and $d_{G_i}(u,v) \le d_{G_i^T}(u,v)$. Then by Lemma \ref{lem:diam_upper_bound}, we have $d(u,v) \ge \min\{ d_{G_i^S}(u,v), d_{G_i^T}(u,v)\} \ge d_{G_i}(u,v)$.
\end{proof}

\begin{lemma}\label{lem:diam_lower_bound_combined}
For every $i$, and every pair of vertices $u,v \in V_i$, $d_{G_i}(u,v) \ge d(u,v)-4D'$.
\end{lemma}
\begin{proof}
    Consider the shortest path from $u$ to $v$ in $G_i$. If this path does not contain both $w_i^S$ and $w_i^T$, then this path exists in $G_i^S$ or $G_i^T$, and thus we can directly apply Lemma \ref{lem:diam_lower_bound} to get $d_{G_i}(u,v) \ge d_{G_i^S}(u,v) \ge d(u,v)-2D'$ 
    %\ge d(u,v)-4D'$
    , or $d_{G_i}(u,v) \ge d_{G_i^T}(u,v) \ge d(u,v)-2D'$.
    %\ge d(u,v)-4D'$.
    
    Otherwise, the shortest path from $u$ to $v$ contain both $w_i^S$ and $w_i^T$. Such path can only be one of the following two forms: 
    \begin{itemize}
        \item $u \rightarrow w_i^S \rightarrow x \rightarrow w_i^T \rightarrow v$ for some vertex $x \in V_i$. The first half $u \rightarrow w_i^S \rightarrow x$ is contained in $G_i^S$, so we can apply Lemma \ref{lem:diam_lower_bound} to get $d_{G_i}(u, x) = d_{G_i^S}(u,x) \ge d(u,x)-2D'$; similarly, the second half $x \rightarrow w_i^T \rightarrow v$ is contained in $G_i^T$ so $d_{G_i}(x, v) \ge d(x,v) -2D'$. In total, $d_{G_i}(u,v)=d_{G_i}(u,x)+d_{G_i}(x,v)\ge (d(u,x)-2D')+(d(x,v)-2D') \ge d(u,v)-4D'$.
        \item $u \rightarrow w_i^T \rightarrow x \rightarrow w_i^S \rightarrow v$ for some vertex $x \in V_i$. We can similarly split this path to two halves, and apply the same analysis as the previous case to get $d_{G_i}(u,v) \ge d(u,v)-4D'$. 
    \end{itemize}
\end{proof}

We are now ready to prove our approximation ratio guarantee: $D / (4\ell-1) \le \tilde{D} \le D$.
%\begin{proof}[Proof of Theorem \ref{thm:min_diam_param}]
We prove the result inductively. When $\ell=1$, it is exactly Theorem \ref{thm:min_diam_no_param}. Now assume it is true for $\ell-1$, and we will prove it for $\ell$. 

%Suppose $D$ is the true Min-Diameter. 
Clearly $D'\le D$ because $D'$ is the min-eccentricity of a vertex. By induction, the $(4\ell-5)$-approximation for the min-diameter of $G_i$ never exceeds the true min-diameter of $G_i$. Then by Lemma \ref{lem:diam_upper_bound_combined}, the min-diameter of $G_i$ does not exceed the min-diameter of $G$. 
%It never exceeds $D$ because $\max_{i, u \in V_i, v \in V_i} \min \left( d_{G_i}(u,v),
%d_{G_i}(v, u)\right)\le \max_{i, u \in V_i, v\in V_i} d_{min}(u,v) \le D$ by . 
Therefore, we never over estimate the min-diameter. 

If $s^*\in W$ or $t^*\in W$, then since we run Dijkstra from all vertices in $W$ we have $D'=D$. So assuming that $s^*,t^*\notin W$, we have two cases.\\
\noindent {\bf Case 1:} $s^*$ and $t^*$ are not in the same vertex set $V_i$. By Lemma~\ref{lem:sampling}, property 2, there exists $w \in W$ such that one of $s^*$ and $t^*$ is in $S_{w}$ and the other is in $T_{w}$, so by Lemma~\ref{lem:splitting_v}, $\eps(w) \ge D/2$. Since $D' \ge \eps(w)$, we have $D' \ge D/2$. \\
\noindent {\bf Case 2:} $s^*$ and $t^*$ are in the same vertex set $V_i$ for some $i$. 
If $D' \ge D/(4\ell-1)$, $D'$ is already a good approximation.
So assume $D' < D/(4\ell-1)$.
%then $s^*,t^*\in V_i$ for some $i$, because otherwise by Lemma \ref{lem:splitting_v} $D' \ge D/2$, a contradiction. 
By Lemma \ref{lem:diam_lower_bound_combined}, $\min\{d_{G_i}(s^*,t^*),d_{G_i}(t^*,s^*)\} \ge d_{min}(s^*,t^*)-4D' = D-4D'$. Since we calculate a $(4\ell-5)$-approximation of $G_i$'s min diameter, our estimate is at least $$(D-4D')/(4\ell-5) \ge (D-4(D/(4\ell-1)))/(4\ell-5) = D/(4\ell-1)$$
%Thus, we have a $(4\ell-1)$-approximation. 
\paragraph*{Runtime analysis}
%Next we perform the runtime analysis for different steps:
%\begin{itemize}
    It takes $\tilde{O}(mn^{1/(\ell+1)})$ time to perform the partitioning from Lemma~\ref{lem:sampling} and to perform Dijkstra's algorithm from all $w\in W$ since 
    %It takes $\tilde{O}(mn^{1/(\ell+1)})$ time to compute distances between all $w \in W$ and other nodes as 
    $|W| = O(n^{1/(\ell+1)})$.
    For all $i$, the number of vertices in $G_i$ is $\abs{V_i}+2=O(n^{\ell/(\ell+1)})$ with high probability by Lemma~\ref{lem:sampling}, property 1, and the number of edges is $m_i+O(n^{\ell/(\ell+1)})$ where $m_i$ is the number of edges in the graph induced by $V_i$. By induction, it takes $\tilde{O}\left((m_i+n^{\ell/(\ell+1)}) \left(n^{\ell/(\ell+1)}\right)^{1/\ell}\right)$ time to compute a $(4\ell-5)$-approximation of Min-Diameter of $G_i$ for each $i$. Summing over all $i$ gives us $\tilde{O}(mn^{1/(\ell+1)})$. 
    
%\end{itemize}

 %So overall, the time complexity is $\tilde{O}(mn^{1/(\ell+1)})$.
 Note that we apply Lemma~\ref{lem:sampling} at most $poly(n)$ times in the recursion and this the only randomization so the whole algorithm works with high probability. 
%\end{proof}

%\begin{theorem}\label{thm:diam_linear}
%    There exists a randomized algorithm in $\tilde{O}(n+m)$ time that can correctly compute a $O(\log n)$-approximation of Min-Diameter with high probability.  
%\end{theorem}

%\begin{proof}[Proof Sketch]
%The proof follows the same method as the proof for Theorem \ref{thm:min_diam_param}. The only difference is the number of vertex sets we split the whole graph into at each level of recursion. 
%At each level of recursion, we choose one single vertex that splits the whole graph to two vertex sets, and use the same algorithm to construct the smaller graphs and compute the estimate. Recursing for $O(\log n)$ levels completes the algorithm. 
%\end{proof}

\section{Min-Radius Algorithm}\label{section:rad}

\begin{theorem}
\label{thm:radius}
For any constant $\delta$ with $1>\delta>0$, there is an $\tilde{O}(m\sqrt n/\delta)$ time randomized algorithm that, given a directed weighted graph $G=(V,E)$ with weights positive and polynomial in $n$, can output an estimate $R'$ such that $R\le R'\le (3+\delta)R$ with high probability, where $R$ is the min-radius of the $G$. 
\end{theorem}
\begin{proof}
We fix a value $r$ and our algorithm either certifies that $R>r$ or $R\leq 3r$. Then by a binary search argument we get a $(3+\delta)$-approximation as follows. Let $\delta'=\delta/3$. Starting from $r=1$, we run the algorithm and increase $r$ for each run. If the output of the algorithm is that $R\le 3r$, then we stop. Otherwise (if $R>r$), we run the algorithm with the new value $r_{new}=(1+\delta')r$. This contributes a multiplicative factor of $\log_{1+\delta'}R=\tilde{O}(1/\delta)$ to the total runtime. Suppose that for some value of $r$ we have $R\le 3r$. So from the previous run of the algorithm, we know that $R>r/(1+\delta')$. Letting $R'=3r$, we have $R\le 3r=R'< 3(1+\delta')R=(3+\delta)R$, which means that $R'$ is a $(3+\delta)$-approximation. Now we present the algorithm. 

\subsubsection*{Algorithm Step 1: Preliminaries}
Let $c$ be the min-center (which is unknown). First we remove all the edges with weight more than $r$, because if $R\le r$, this removal does not change the min-radius. Then we sample a set $W$ of $\sqrt{n}$ vertices according to Lemma \ref{lem:sampling}. For every vertex $v\in W$, we run Dijkstra's algorithm from and to $v$ to obtain the min-distance between $v$ and all other vertices. If there exists a vertex $v\in W$ with $\eps(v)\leq 3r$, we have certified that $R\leq 3r$ so we are done. 
%then in any shortest path between $c$ and any other node in the min-distance sense, these edges are not used. 

\subsubsection*{Algorithm Step 2: Constructing the ``far graph"}
 Now we can assume that for each $v\in W$, $\eps(v)>3r$. We say that a pair of vertices is {\it far} if their min-distance is more than $2r$, and let the \emph{far graph} $G_{far}$ be an undirected unweighted graph on $V$ defined as follows: for each $u\in W$ and $v\in V$, $(u,v)$ is an undirected edge if $u$ and $v$ are far. We partition $W$ based on the connected components of $G_{far}$. Specifically, for all $i$ define $Z_i$ to be the $i^{th}$ connected component of $G_{far}$ which contains at least one vertex in $W$. Let $F_i = W \cap Z_i$, note that $F_i$ is non-empty.

%$F_1,\ldots, F_{l}$ for some $l \le |W| = \sqrt{n}\ln{n}$. be a partition of $W$ based on which connected component of 

%todo go through entire document search for distance and make sure it's clear whether it's distance or min-distance.
%todo recall defn of S_f and T_f?
\subsubsection*{\emph{Analysis Step 2}} 

Remember that we defined $S_U=\bigcap_{v\in U}S_v$ and $T_U=\bigcap_{v\in U}T_v$.

By constructing $G_{far}$, we prune the set of candidate min-centers, as specified in the following lemma.
%First note that $c$ is not far from any vertex, so it forms a single connecting component $F_c=\{c\}$. 
\begin{lemma}
\label{lem:SorT}
If $R\le r$, then for any $F_i$ either $c\in S_{F_i}$ or $c\in T_{F_i}$. 
\end{lemma} 
\begin{proof}
First note that we have $S_{F_i}\cup T_{F_i}\neq V\setminus F_i$. We know that $c \not\in F_i$ as $F_i \subseteq W$. By way of contradiction, assume that there are two vertices $u,v \in F_i$ such that $c\in S_u\cap T_v$. Consider a path in $G_{far}$ from $u$ to $v$. There must be a pair of adjacent vertices $(u',v')$ on the path such that $c\in S_{u'}\cap T_{v'}$. Then, by definition, $u'$ and $v'$ are far (with respect to the original graph $G$). Since $c\in S_{u'}\cap T_{v'}$, we have $d(v',c)\le r$ and $d(c,u')\le r$, so by the triangle inequality $d(v',u') \le 2r$. Thus $u'$ and $v'$ are not far, a contradiction.
\end{proof}
%todo define in prelims section subscript for eps and d

 %We are going to merge some of these components in the next step, to get another division of the set $W$.

\subsubsection*{Algorithm Step 3: Defining a DAG-like structure}

\paragraph*{a) Constructing the ``close graph"} The purpose of constructing the close graph is that it allows us to either perform Dijkstra's algorithm from some additional vertices and obtain a good estimate (see step b), or ``merge" some connected components of the far graph to further prune the set of vertices that could be the min-center (see step c). 
%todo explain above better?
 The \emph{close graph} $G_{close}$ is an unweighted directed graph with one vertex $f_i$ for each $F_i$.
%For each connected component $F_i$, let the node $f_i$ correspond to it. 
%Define the directed unweighted graph $H$ on $f_1,\ldots, f_{l}$ as follows: For each $1\le i \neq j \le l$, let $(f_i,f_j)$ be an edge if for some $u\in F_i$ and some $v\in F_j$, $d(u,v)\le5r$ (equivalently we are substituting the paths of length at most $5r$ by a single directed edge).
For all $i$ and $j$, let $(f_i,f_j)$ be an edge in $G_{close}$ if for some $u\in F_i$ and some $v\in F_j$, $d(u,v)\le 5r$.

\paragraph*{b) Additional Dijkstra} We now perform Dijkstra's algorithm from some additional vertices, which are carefully chosen so that either we find a vertex with small min-eccentricity and are done in this step, or we can define a DAG-like structure in the graph (step c). We compute the strongly connected components (SCCs) of $G_{close}$. For each SCC $Q=(V_Q,E_Q)$, find $E'_Q\subseteq E_Q$ with $|E'_Q|\leq 2|V_Q|$ such that $Q'=(V_Q,E'_Q)$ is strongly connected; it is simple to show that such an $E'_Q$ exists and we include the proof in the appendix for completeness (Lemma~\ref{lem:SCC}). Let $E'=\cup_{Q}E'_Q$. Note that every edge $e\in E'$ corresponds to a path $P_e$ of length at most $5r$ in the original graph $G$. For each $e\in E'$, find an ordered set $V_e$ of at most 9 vertices on $P_e$ that divide $P_e$ into subpaths of length at most $r$; it is simple to show that such a $V_e$ exists and we include the proof in the appendix for completeness (Lemma~\ref{lem:divideTo_r_sections}). We run Dijkstra's algorithm from every vertex in $V_e$ and if we find a vertex $v$ with $\eps(v)\leq 3r$ then we are done.
%from all of these vertices Let $H$ and call the union of these subgraphs over all SCCs by $H'$. Note that each edge of $H'$ corresponds to a path of length at most $5r$. Since there are no edges of weight more than $r$ in the graph, using Lemma \ref{lem:divideTo_r_sections} divide that path into sections of length at most $r$ using at most $9$ internal vertices of the path, run Dijkstra from each dividing vertex. If we find a vertex $v$ with $\eps(v)\leq 3r$, then we are done. 

\paragraph*{c) Constructing the DAG of ``supercomponents"} 
%todo? Maybe this section can be written more clearly?

%Assuming we are not done by the previous step, If none of the vertices $\le 3r$, merge all far graph connected components corresponding to vertices in $Q$.  
Let $H$ be the DAG created by contracting every strongly connected component of $G_{close}$ into a single vertex. That is, there is an edge from $u$ to $v$ in $H$ if the strongly connected component $v$ is reachable from the strongly connected component $u$. Let $k$ be the number of vertices in $H$; we number the vertices in $H$ from 1 to $k$ according to a topological ordering. For each $j\in [k]$, we merge the set of $F_i$'s represented by vertex $j$ in $H$ into a \emph{supercomponent} $W_j$. Formally, if we define $F_u$ to be the connected component of $G_{far}$ that contains $u$, a vertex $u\in W$ is in supercomponent $W_j$ if $f_u$ is in the strongly connected component of $H$ represented by vertex $j$. 

\paragraph*{d) Fitting the remaining vertices into the DAG structure} In the previous step, we defined a DAG-like structure on the vertices in $W$. Now we place the rest of the vertices into this structure. We partition the rest of the vertices based on whether they could potentially be the min-center. We define the vertex sets $C$ and $B$ next and in the analysis we prove that $c\in C$ (among other properties of $C$ and $B$). 
%todo placement of figures
We will use the following notation: for any distance $d>0$, let $S_v^d=\{u\in S_v: d(u,v)\le d\}$ and let $T_v^d=\{u\in T_v: d(v,u)\le d\}$. Remember that for any set $U$ of vertices, we defined $S_U^d=\bigcap_{v\in U}S_v^d$, and $T_U^d=\bigcap_{v\in U}T_v^d$.
\begin{itemize}
    \item For $i=1,\ldots, k+1$, let $v\in C_i$ if for all $j<i$, $v\in T_{W_j}^{2r}$ and for all $j\ge i$, $v\in S_{W_j}^{2r}$. Let $C=\cup_{i=1}^{k+1} C_i$.
    \item For $i=2,\ldots, k+1$, let $v\in B_i$ if $v\notin C$ and $i$ is the largest integer for which $v\in T_{W_{i-1}}^{2r}$. Let $v\in B_1$ if there is no such $i$ and $v\notin C$. Let $B=\cup_{i=1}^{k+1} B_i$.
\end{itemize}

%So the set $W$ is divided into subsets $W_1,\ldots, W_k$ which are the unions of connected components of the far graph, and we call them {\it supercomponents}. 

\subsubsection*{\emph{Analysis Step 3}}

Figure \ref{fig:graph-structure} shows a summary of the structure of the graph which we will describe in the following observations and lemmas. 
 \begin{figure}[h]
   \centering
     \includegraphics[width=0.5\textwidth]{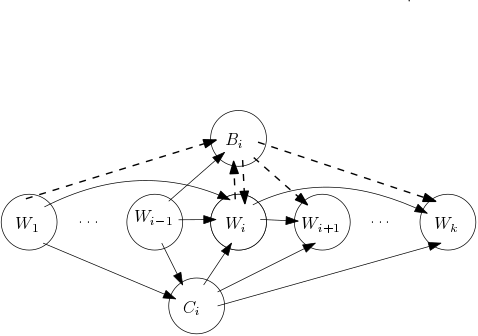}
     \caption{The graph structure for the sets $W_i$, $B_i$ and $C_i$. Solid lines are paths of length at most $2r$ between any member of the outgoing set to any member of the incoming set. Dashed lines are paths of length at most $2r$ which might not exist between all pairs, which is expressed more accurately in Lemma \ref{lem:propB}.}
 	\label{fig:graph-structure}
 \end{figure}
 
We first observe two important properties of supercomponents:
\begin{observation}\label{obs:backedge}
For every pair of vertices $v_i \in W_i$ and $v_j \in W_j$ with $i<j$, $d(v_j,v_i)>5r$.
\end{observation}

This is true because if $d(v_j,v_i)\le 5r$, then there is an edge from $f_j$ to $f_i$ in $G_{close}$, so there is an edge from $j$ to $i$ in $H$. Since $i<j$, this contradicts the topological ordering of $H$.

\begin{observation}
\label{obs:completegraph}
For every pair of vertices $v_i \in W_i$ and $v_j \in W_j$ with $i<j$, $v_i\in S_{W_j}^{2r}$ and $v_j\in T_{W_i}^{2r}$.
\end{observation}

This is true because $v_i$ and $v_j$ are in different $F_k$'s since $W_i$ and $W_j$ are collections of disjoint sets of $F_i$'s. So $v_i$ and $v_j$ are not far i.e. $d_{min}(v_i, v_j) \le 2r$ and by Observation \ref{obs:backedge} we know that $d(v_j,v_i)>5r>2r$, so it must be that $d(v_i,v_j)\le 2r$. Since this is true for all vertices $v_j\in W_j$, we have $v_i\in S_{W_j}^{2r}$. Similarly, $v_j\in T_{W_i}^{2r}$.

We now prove a refinement of Lemma~\ref{lem:SorT} where we consider supercomponents instead of far graph components. This further prunes the vertices that could potentially be the min-center.

\begin{lemma}
\label{lem:closegraph}
%close graph lemma
If $R\le r$, then for each $i=1,\ldots,k$, either $c\in S_{W_i}$ or $c\in T_{W_i}$.
\end{lemma}

\begin{proof}
Fix $i$ and suppose by way of contradiction that there are nodes $u,v\in W_i$ such that $c\in S_u\cap T_v$. By Lemma~\ref{lem:SorT}, $u$ and $v$ must be in different $F_i$'s say $F_u$ and $F_v$. %Now we show that we can assume that $d(u,v)<5r$:

 Recall that by the definition of a supercomponent, $f_u$ and $f_v$ are in the same strongly connected component of $G_{close}$. So there is a path $P$ from $f_u$ to $f_v$ in $G_{close}$ such that all of its edges are in $E'$. By Lemma~\ref{lem:SorT} Since $c\in S_u\cap T_v$, we have that $c\in S_{F_u}\cap T_{F_v}$. So there are two consecutive nodes $f_{j}$ and $f_{j'}$ on $P$ (in that order) such that $c\in S_{F_{j}}\cap T_{F_{j'}}$. %Take two nodes $u'\in F_j$ and $v'\in F_{j'}$. Then $d(u',v')<5r$, and $c\in S_{u'}\cap T_{v'}$.

%So we assume that $d(u,v)<5r$. 
Recall that each edge $e\in E'$ corresponds to a path $P_e$ of length at most $5r$ in the original graph. Let $e$ be the edge $(f_j,f_{j'})$ and consider $P_e$ and $V_e$, where $V_e$ is the set of vertices that divides $P_e$ into subpaths of length at most $r$. Since the endpoints of $P_e$ are in $F_j$ and $F_{j'}$ respectively, there exists a pair of vertices $u',v'$ consecutive in $V_e$ (in that order) such that $c\in S_{u'}\cap T_{v'}$. We note that $d(u',v')\leq r$.

%old version:
%So we assume that $d(u,v)<5r$. 
%Now by Lemma \ref{lem:divideTo_r_sections}, there are at most $9$ nodes $v_1,\ldots,v_9$ on the path from $u$ to $v$ that divide the path into subpaths of length at most $r$. So there is some $i\in\{0,\ldots,9\}$ such that $c\in T_{v_i}$ and $c\in S_{v_{i+1}}$, where $v_0=v$ and $v_{10}=u$.

%Now we claim that $v_i$ has min-eccentricity at most $3r$: This is because $d(v_i,c)\le r$ and $d(c,v_i)\le d(c,v_{i+1})+d(v_{i+1},v_i) \le 2r$. So considering any vertex $w\in V$, if $d(c,v)\le R$, then $d(v_i,v)\le d(v_i,c)+d(c,v) \le 2r$. If $d(v,c)\le R$, then $d(v,v_i)\le d(v,c)+d(c,v_i)\le 3r$.

Now we claim that $\eps(v')\leq 3r$. This is because $d(v',c)\le r$ and $d(c,v')\le d(c,u')+d(u',v') \le 2r$. Consider an arbitrary vertex $w\in V$. Either $d(c,w)\le R$ or $d(w,c)\le R$. If $d(c,w)\le R$ then $d(v',w)\le d(v',c)+d(c,w) \le 2r$. If $d(w,c)\le R$, then $d(w,v')\le d(w,c)+d(c,v')\le 3r$. In this case, the algorithm would have stopped after step 3b.

\end{proof}

We now prove that $c\in C$, which further prunes the vertices that could potentially be the min-center.

\begin{lemma}
\label{lem:propC}
If $R\le r$, then $c\in C$.
\end{lemma}
\begin{proof}
By Lemma~\ref{lem:closegraph}, either $c\in S_{W_i}$ or $c\in T_{W_i}$. Since $c$ is the min-center and $R\leq r$, if $c\in S_{W_i}$ then $c\in S_{W_i}^{2r}$, and similarly if $c\in T_{W_i}$ then $c\in T_{W_i}^{2r}$. We claim that for each $i<j$, $S_{W_i}^{2r}\cap T_{W_j}^{2r} = \emptyset$, which completes the proof. Suppose otherwise and let $i$ and $j$ be such that $i<j$ and there is a vertex $v\in S_{W_i}^{2r}\cap T_{W_j}^{2r}$. Then for every vertex $v_i\in W_i$ and $v_j\in W_j$, $d(v_j,v)\le 2r $ and $d(v,v_i)\le 2r$, so $d(v_j,v_i)\le 4r$. This contradicts Observation \ref{obs:backedge}.
%because otherwise there is a path of length at most $4r$ from $W_j$ to $W_i$, contradicting the DAG ordering. Using this fact together with Lemma , the claim follows.
\end{proof}

Now we prove that the vertices in $B$ fit into the DAG structure in a similar but weaker sense than the vertices in $C$:

%\begin{enumerate}[label=\roman*]
  %  \item \label{twofarcomp}
     
    %For any two nodes $w_i\in W_i$ and $w_j\in W_j$ for some $i<j$, since $w_i$ and $w_j$ are in different connected components in the far graph, there is a path of length at most $2r$ from $w_i$ to $w_j$.
  %  \item \label{backedgeprop} For each $i<j$, $S_{W_i}^{2r}\cap T_{W_j}^{2r} = \emptyset$, i.e. there is no vertex $v$, where $d(W_j,v)\le 2r $ and $d(v,W_i)\le 2r$, because otherwise there is a path of length at most $4r$ from $W_j$ to $W_i$, contradicting the DAG ordering. 
%\end{enumerate}

\begin{lemma}
\label{lem:propB}
    Consider a node $v\in B_i$. Then for all $z\ge  i$ except for at most two values, we have $v\in S_{W_z}^{2r}$. And for all $z\le i$ except for at most two values, we have $v\in T_{W_z}^{2r}$.

\end{lemma}

\begin{proof}
 We first observe that there is at most one $j$ such that $v$ is far from some vertex in $W_j$. This is because if $v$ were far from two vertices $u,w$ in different supercomponents, then $G_{far}$ would contain the edges $(u,v)$ and $(w,v)$ making $u$ and $w$ in the same connected component of $G_{far}$, and thus in the same supercomponent. We fix $j$ and consider two cases:
 
\noindent {\bf Case 1:} Suppose by way of contradiction that for some node $w\in W_z$ for some $z<i,$ $z\neq j$, we have $v\in S_{w}^{2r}$. We know that $z<i-1$, since by definition of $B_i$, we have $v\in T_{W_{i-1}}^{2r}$. Let $w'\in W_{i-1}$ be an arbitrary node, then $d(w',w)\le d(w',v)+d(v,w)\le 2r+2r< 5r$, a contradiction to Observation~\ref{obs:backedge}.
%to property \ref{backedgeprop}. 

\noindent {\bf Case 2:} Now suppose that for some node $w\in W_z$ for some $z>i,$ $z\neq j$, we have $v\in T_{w}^{2r}$. We will show that $j=i$ and $z=i+1$; that is, for all $z'\ge i+2$, we have that $v\in S_{W_{z'}}^{2r}$. If there is some node $w'\in W_i$ such that $v\in S_{w'}^{2r}$, then $d(w,w')\le d(w,v)+d(v,w')\le 2r+2r<5r$, a contradiction to Observation~\ref{obs:backedge}. Assume that there is no such $w'$ i.e. $d(v,w')>2r$ for all $w'\in W_i$. Then for every node $w'\in W_i$, either $v$ and $w'$ are far or $d(w',v)\leq 2r$. If for all $w'\in W_i$, $d(w',v)\leq 2r$, then $v\in T_{W_{i}}^{2r}$, which cannot happen since by the definition of $B_i$, $i$ is the biggest integer that $v\in T_{W_{i-1}}^{2r}$. Thus, $v$ is far from some vertex in $W_i$ so we have that $j=i$.  If $z>i+1$, then by definition of $B_i$ there is some vertex $u\in W_{i+1}$ such that $v\in S_{u}^{2r}$. So $d(w,u)\le d(w,v)+d(v,u)\le 2r+2r<5r$, a contradiction to Observation~\ref{obs:backedge}. So it must be that $z=i+1$. So for all $z'\ge i+2$, we have that $v\in S_{W_{z'}}^{2r}$.
\end{proof}

We have observed stronger properties than Lemma~\ref{lem:propB} for vertices $v\in W_i$ (Observation \ref{obs:completegraph}) and $v\in C_i$ (by definition), so we have the following corollary. 

\begin{corollary}
\label{cor:Bproperty_extension}
Lemma \ref{lem:propB} is true for all $v\in B_i\cup C_i\cup W_i$. Moreover, for such $v$'s, we have $v\in T_{W_{i-1}}^{2r}$.
\end{corollary}

\subsubsection*{Algorithm Step 4: Partial search} From each of the potential min-centers, we will run Dijkstra's algorithm on a small subgraph of $G$. For each $i=1,\ldots,k+1$, let $G_i$ be the subgraph of $G$ induced by $W_{i-6}\cup\ldots\cup W_{i+3}\cup B_{i-6}\cup\ldots\cup B_{i+3}\cup C_{i-5}\cup \ldots\cup C_{i+3}$. Define $\bar{C}_i$ to be the set of nodes $v \in C_i$ such that $v$ is within min-distance $r$ from all vertices in $W$ (we know this set of nodes because we have already run Dijkstra's algorithm from and to every vertex in $W$). For every vertex $v$ in $\bar{C}_i$, run Dijkstra's algorithm from $v$ with respect to the graph $G_i$. If $v$ is within min-distance $r$ from all nodes in $U_i=C_i\cup B_{i-2}\cup B_{i-1}\cup B_i$, we will show that $R\leq 3r$. If there is no such $v$, we will show that $r<R$.

\subsubsection*{\emph{Analysis Step 4}} 
%This step of the algorithm takes $O(m\sqrt{n})$ time. (amortized argument TODO).

The following two claims prove that our algorithm either certifies that $R>r$ or $R\leq 3r$.

%The correctness of the entire algorithm follows from Claims~\ref{claim:sp_inside} and \ref{claim:approx_center_valid}:

\begin{claim}
\label{claim:sp_inside}
 For some $i$, if $c\in \bar{C}_i$ and $R\le r$, then for all $u\in U_i$, the min-distance between $c$ and $u$ with respect to $G_i$ is at most $r$. %$d_{min}(u,v)$ all its shortest paths with vertices in $U_i$ are inside $G_i$.
\end{claim}
\begin{claim}
\label{claim:approx_center_valid}
If a vertex $v\in \bar{C}_i$ is within min-distance $r$ from all vertices in $U_i$ with respect to the graph $G_i$, then $R \le \eps(v)\leq 3r$. %it has min-eccentricity at most $3r$ in $G$.
\end{claim}
%todo include figure for the following pf
\begin{proof}[Proof of Claim \ref{claim:sp_inside}] 
We will prove something slightly stronger: for all $i$ any shortest path in $G$ between two nodes $u,u'\in U_i$ that has length at most $r$ is completely contained in $G_i$. 

By way of contradiction, suppose that the shortest path $P$ from $u$ to $u'$ is not completely in $G_i$. Define $V_{r}$ and $V_{l}$ to be sets of nodes on the right and left of $G_i$ respectively, i.e. $V_{r}=W_{i+4}\cup \ldots \cup W_{k}\cup B_{i+4}\cup \ldots \cup B_{k+1}\cup C_{i+4}\cup \ldots \cup C_{k+1}$ and $V_{l}=W_{1}\cup \ldots \cup W_{i-7}\cup B_{1}\cup \ldots \cup B_{i-7}\cup C_{1}\cup \ldots \cup C_{i-6}$. 

First suppose that $P$ contains some node $v_r\in V_r$. There is some $j>i+3$ such that $v_r\in B_j\cup C_j\cup W_j$. So by Corollary \ref{cor:Bproperty_extension}, $v_r\in T_{W_{j-1}}^{2r}$.
Furthermore, Corollary \ref{cor:Bproperty_extension} implies that there is some $j'\in \{i,i+1,i+2\}$ such that $u'\in S_{W_{j'}}^{2r}$. Pick $w_{j'}\in W_{j'}$ and $w_{j-1}\in W_{j-1}$. We have that $d(w_{j-1},w_{j'})\le d(w_{j-1},v_r)+ d(v_r,u')+d(u',w_{j'})\le 2r+r+2r = 5r$. Since $j-1>j'$, this contradicts Observation~\ref{obs:backedge}. This case is shown in Figure \ref{fig:claim1case1}.

 \begin{figure}
   \centering
     \includegraphics[width=\textwidth]{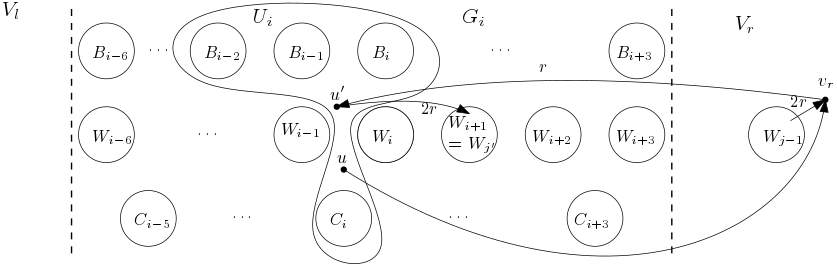}
     \caption{First case in Claim \ref{claim:sp_inside} where the path $P$ from $u$ to $u'$ passes through some vertex $v_r\in V_r$. In this figure $j'=i+1$. The upper bound on the weight of each part of the path from $W_{j-1}$ to $W_{j'}$ is written on the edges.}
 	\label{fig:claim1case1}
 \end{figure}

Now suppose that $P$ contains some node $v_l$ in $V_l$. The argument in this case is symmetric to the previous case. Since $u\in U_i$, there is some $j\in \{i-2,i-1,i\}$ such that $u\in B_j\cup C_j$, and hence by Corollary \ref{cor:Bproperty_extension}, $u\in T_{W_{j-1}}^{2r}$. Furthermore, Corollary \ref{cor:Bproperty_extension} implies that there is at least one value $j'\in \{i-5,i-4,i-3\}$ such that $v_l\in S_{W_{j'}}^{2r}$. Pick $w_{j'}\in W_{j'}$ and $w_{j-1}\in W_{j-1}$. 
We have that $d(w_{j-1},w_{j'})\le d(w_{j-1},u)+d(u,v_l)+d(v_l,w_{j'})\le 2r+r+2r=5r$.
%previous version below
%We have that $d(w_{j'}, w_{j-1}) \le d(w_{j'},u)+d(u,v_l), d(v_l,w_{j-1})\le 2r+r+2r=5r$. 
Since $j'<j-1$, this contradicts Observation~\ref{obs:backedge}.
\end{proof}

\begin{proof}[Proof of Claim \ref{claim:approx_center_valid}]
We show that for any node $u\in V$ we have $d_{min}(u,v)\le 3r$. We have $3$ cases:

\noindent{\bf Case 1:} $u\in W$: From the definition of $\bar{C}_i$, we know that $v$ has min-distance at most $r$ to all vertices in $W$.

\noindent{\bf Case 2:} $u\in C_j$ for some $j=1,\ldots,k+1$. If $j=i$, then $u\in U_i$ so we know that $d_{min}(u,v)\le r$. If $j> i$, then pick some vertex $w_i\in W_i$. By the definition of $C_i$ and $C_j$ we know that $d(v,u)\le d(v,w_i)+d(w_i,u)\le r+2r = 3r$. Similarly, if $j\le i-1$, pick some vertex $w_{i-1}\in W_{i-1}$. Then $d(u,v)\le d(u,w_{i-1})+d(w_{i-1},v)\le 2r+r=3r$.

\noindent{\bf Case 3:} $u\in B_j$ for some $j=1,\ldots, k+1$. If $j\in \{i-2,i-1,i\}$, then since $v\in \bar{C}_i$ we know that $d_{min}(u,v)\le r$. 
%$v$ is within min-distance $r$ from $u$ in $G_i$, and since $G_i$ is a subgraph of $G$, they are within min-distance $r$ in $G$. 
So first suppose that $j\le i-3$. Then by Lemma \ref{lem:propB}, there is at least one $j'\in \{j+1,\ldots,i-1\}$, such that $u\in S_{W_{j'}}^{2r}$. Pick some node $w_{j'}\in W_{j'}$. So by the definition of $C_i$ we have that $d(u,v)\le d(u,w_{j'})+d(w_{j'},v)\le 2r+r=3r$. 
Now suppose that $j\ge i+1$. Then by definition of $B_j$ we know that $u\in T_{W_{j-1}}^{2r}$. Pick some vertex $w_{j-1}\in W_{j-1}$. Since $j-1\ge i$ and by the definition of $C_i$, we have that $d(v,u) \le  d(v,w_{j-1})+d(w_{j-1},u)\le r+2r=3r.$

\end{proof}

%{\it rest of the Analysis Step 3.} 
%If $R\le r$, by Claim \ref{claim:sp_inside} we know that the algorithm will output some node, and by Claim \ref{claim:approx_center_valid} we see this node has min-eccentricity at most $3r$. So if the algorithm doesn't return any node, it means that $R>r$.

%todo add to runtime the fact that C_is are small.

\subsubsection*{Runtime analysis}
We analyze the running time of each step.

\noindent{\bf Step 1, preliminaries:} $\tilde{O}(m\sqrt{n})$. This is because each Dijkstra in Step 1 takes $\tilde{O}(m)$ time and $|W|=\sqrt{n}$.

\noindent{\bf Step 2, constructing the “far graph”:} $\tilde{O}(n\sqrt{n})$. Each edge in the far graph has at least one endpoint in $W$, and so the construction of the far graph takes $O(n\sqrt{n})$ time. Note that the existence of each edge in the far graph was determined in Step 1.

\noindent{\bf Step 3, defining a DAG-like structure:} 

{\bf a, constructing the “close graph”: } $\tilde{O}(n\sqrt{n})$. This is because the connected components of the far graph can be determined in $O(n\sqrt{n})$ time since it has that many edges. The number of components containing a node in $W$ are not more than $|W|$, and so the close graph which is on at most $|W|$ nodes can be constructed in time $O(|W|^2)=O(n)$.

{\bf b, additional Dijkstra:} $\tilde{O}(m\sqrt{n})$. This is because by Lemma \ref{lem:SCC}, the number of vertices we run Dijkstra from in SCC $Q$ of the close graph is at most $9|E_Q|\le 18|V_Q|$. Since the number of vertices in close graph is at most $|W|$, we run Dijkstra from at most $18|W|=\tilde{O}(\sqrt{n})$ vertices. Also running the algorithm of Lemma \ref{lem:SCC} takes $O(|E_Q|)=O(|V_Q|^2)$ for each SCC $Q$, which takes $O(|W|^2)=\tilde{O}(n)$ time in total.

{\bf c, constructing  the  DAG  of  “supercomponents”:} $\tilde{O}(n)$. This is because $H$
has at most $|W|$ vertices, so obtaining the DAG ordering of $H$ takes at most $|W|^2=\tilde{O}(n)$ time.

{\bf d, fitting the remaining vertices into the DAG structure:} $\tilde{O}(n\sqrt{n})$. For each vertex in $V$, it takes $O(|W|)$ time to see which set it belongs to, since it only depends on its distances to and from the vertices in $W$. 

\noindent{\bf Step 4, partial search:} $\tilde{O}(m\sqrt{n})$. The Dijkstras ran in $G_i$ take $\tilde{O}(m_{i}|C_i|)$ time, where $m_i$ is the number of edges with at least one endpoint in $G_i$. By Lemma \ref{lem:sampling}, property 3, with high probability $|C_i|=O(\sqrt{n})$. We know that $ \bar{C}_i\subseteq C_i$ and so the running time of this step is $O(\sqrt{n}\sum_{i=1}^{k+1} m_i)$. Now since each node is in at most $10$ $G_i$s, we have that each edge is also in at most $20$ $G_i$s, and hence $\sum_{i=1}^{k+1} m_i \le 20m$.

So overall the algorithm runs in $\tilde{O}(m\sqrt{n})$ time. 

%It remains to see how many steps we need to run the algorithm to complete the argument. 

%\noindent{\bf Binary search running time:} Each time that the algorithm returns that the value $r$ is not big enough ($r<R$), the next value of $r$ is $r(1+\delta')$, where $\delta'=\delta/3$ is a constant. Since $R\le n$, we know that $(1+\delta')^t\le n$, where $t$ is the number of times we run the algorithm. So $t\le \ln{n}/\ln{(1+\delta')}$. For small enough $\delta$ and $\delta'$, we know that $\ln{(1+\delta')}=O(\delta')=O(\delta)$. So the running time of the binary search is $\tilde{O}(m\sqrt{n}/\delta)$.

\end{proof}

%I moved the pfs that were here to the appendix

\section{Min-Eccentricities Algorithm}\label{section:ecc}

%%% skeleton

The min-eccentricities algorithm is similar to the min-radius algorithm. Below we will describe the modifications.

\begin{theorem}
\label{thm:eccntricities}
For any constant $\delta$ with $1>\delta>0$, there is an $\tilde{O}(m\sqrt n/\delta)$ time randomized algorithm, that given a directed weighted graph $G=(V,E)$ with weights positive and polynomial in $n$, can output an estimate $\eps'(s)$ for every vertex $s \in V$ such that $\eps(s)\le \eps'(s)\le (5+\delta)\eps(s)$ with high probability, where $\eps(s)$ is the min-eccentricity of the vertex $s$ in $G$. 
\end{theorem}
\begin{proof}
We fix a value $\rho$ and our algorithm certifies for each $s \in V$ that either $\eps(s)>\rho$ or $\eps(s)\leq 5\rho$ with high probability. Starting from $\rho=1$, we will run the algorithm and increase $\rho$ for each run. We will call the vertices for which we have certified $\eps(s)\leq 5\rho$ for earlier values of $\rho$ as \emph{marked}. Let $\delta' = \delta/5$. Starting from $\rho=1$, we run the algorithm. If the output of the algorithm is that $\eps(s)\le 5\rho$ and $s$ was unmarked, then we will mark $s$ and set $\eps'(s) = 5\rho$. Then, we run the algorithm with the new value $\rho_{new}=(1+\delta')\rho$. Since $\eps(s)\le poly(n)$ for all $s\in V$, this contributes a multiplicative factor of $\log_{1+\delta'}n=\tilde{O}(1/\delta)$ to the total runtime. Suppose that for some value of $\rho$ and for some vertex $s$ we have $\eps(s)\le 5\rho$ and $s$ was unmarked. From the previous run of the algorithm, we know that $\eps(s) >\rho/(1+\delta')$. Then for $\eps'(s)=5\rho$, we have $\eps'(s) \geq \eps(s)$ and $\eps'(s)\le 5(1+\delta')\eps(s)=(5+\delta)\eps(s)$, which means that $\eps'(s)$ is a $(5+\delta)$-approximation of $\eps(s)$. After running the whole algorithm for this value of $\rho$ we will also mark all such vertices $s$. Now we present the algorithm.

Throughout the algorithm $\rho$ behaves analogously to $r$ in the min-radius algorithm. Whenever we say that a certain part of the algorithm is the same we mean that it is same after replacing $r$ by $\rho$. Note that any vertex $s$ with $\eps(s) = \rho$ satisfies the property that its min-distance to any vertex is at most $\rho$. This is analogous to the center vertex $c$ in the Min-Radius algorithm using $r = \rho$.

\subsubsection*{Algorithm Step 1: Preliminaries}
First we remove all the edges with weight more than $\rho$, because if for a vertex $s$ with $\eps(s)\le \rho$, this removal does not change the min-eccentricity of $s$. Then we sample a set $W$ of $\sqrt{n}$ vertices according to Lemma \ref{lem:sampling}. For every vertex $v\in W$, we run Dijkstra's algorithm from and to $v$ to obtain the min-distance between $v$ and all other vertices. This means we know $\eps(v)$ for all $v \in W$ and in particular we know if $\eps(v) > \rho$ or $\eps(v) \le 3\rho$. We use the vertices in $W$ with min-eccentricity less than $3\rho$ to detect vertices with min-eccentricity less than $5\rho$ in the graph.

\subsubsection*{Algorithm Step 2: Constructing the ``far graph"}
The far graph and the $F_i$'s are defined the same way as in the min-radius algorithm. 

%$F_1,\ldots, F_{l}$ for some $l \le |W| = \sqrt{n}\ln{n}$. be a partition of $W$ based on which connected component of 

\subsubsection*{\emph{Analysis Step 2}} 
The purpose of constructing $G_{far}$ is to prune the set of vertices that could potentially have low min-eccentricity.
Next we state a modified Lemma~\ref{lem:SorT}.
\begin{lemma}[Modification of Lemma~\ref{lem:SorT}]
\label{lem:SorT-ecc}
If for a vertex $s \in V\setminus W$, $\eps(s)\le \rho$, then for any $F_i$, either $s\in S_{F_i}$ or $s\in T_{F_i}$. 
\end{lemma}
\begin{proof}
For $s \in F_i$ note that $F_i \subseteq W$ and hence we know $\eps(s)$ and have certified either $\eps(s) > \rho$ or $\eps(s) \le 3\rho$. For the other vertices the proof is analogous to that of Lemma~\ref{lem:SorT} 
\end{proof}

\subsubsection*{Algorithm Step 3: Defining a DAG-like structure}

\paragraph*{a) Constructing the ``close graph"} The purpose of constructing the close graph is that it allows us to perform Dijkstra's algorithm from some additional vertices and certify some vertices as having min-eccentricities $\le 5\rho$. Then  we   ``merge" some connected components of the far graph to further prune the set of vertices that could be having small min-eccentricities (see step c). 
%todo explain above better
$G_{close}$ is defined as in the min-radius algorithm.

\paragraph*{b) Additional Dijkstra} This step of the algorithm diverges from the min-radius algorithm at the end, and hence we state it in full detail. Similar to the min-radius algorithm, we perform Dijkstra's algorithm from some additional vertices, which are chosen so that we detect more vertices with low min-eccentricity and at the end define a DAG-like structure (step c). Recall that we compute the strongly connected components (SCCs) of $G_{close}$. For each SCC $Q=(V_Q,E_Q)$, find $E'_Q\subseteq E_Q$ with $|E'_Q|\leq 2|V_Q|$ such that $Q'=(V_Q,E'_Q)$ is strongly connected (where the existence of $E_Q'$ is shown in Lemma~\ref{lem:SCC}). Let $E'=\cup_{Q}E'_Q$. Recall that every edge $e\in E'$ corresponds to a path $P_e$ of length at most $5\rho$ in the original graph $G$. For each $e\in E'$, find an ordered set $V_e$ of at most 9 vertices on $P_e$ that divide $P_e$ into sections of length at most $\rho$ (see Lemma \ref{lem:divideTo_r_sections}). For each $e\in E'$, we run Dijkstra's algorithm from and to every vertex in $V_e$. This means we know $\eps(v)$ for all $v \in V_e$; and in particular we know whether $\eps(v) > \rho$ or $\eps(v) \le 3\rho$. Now here is the new part of the algorithm in this step: For every consecutive pair of vertices $(a, b)$ in $V_e$ over all $e$ with $\eps(a), \eps(b) \leq 3\rho$ we certify for all $s \in S^{\rho}_b \cap T^{\rho}_a$ that $\eps(s) \le 5\rho$.

%This is indeed true by Lemma~\ref{lem:triangle-lemma-ecc} (in the statement of the lemma, let $c=s$, $\gamma_1=\rho,\gamma_2=2\rho$ and $\gamma_3=3\rho$). Here we see one major difference of this algorithm and the min-radius algorithm; because we cannot upper bound the min-eccentricities of these vertices by better than $5\rho$ we get an approximation factor of $5+\delta$ for min-eccentricities unlike min-radius.
\paragraph*{c) Constructing the DAG of ``supercomponents"} 
%todo? Maybe this section can be written more clearly?

%Assuming we are not done by the previous step, If none of the vertices $\le 3r$, merge all far graph connected components corresponding to vertices in $Q$.  
The graphs $H$, $W_i$'s and the ``supercomponents" are defined as in the min-radius algorithm. 

\paragraph*{d) Fitting the remaining vertices into the DAG structure} In the previous step, we defined a DAG-like structure on the vertices of $W$. Now we place the rest of the vertices into this structure. We partition the rest of the vertices based on whether they haven't been certified to have eccentricity $\le 5\rho$ \emph{and} could potentially have small eccentricity. Vertex sets $C$ and $B$ are defined as in the min-radius algorithm. In the analysis we prove that all vertices which haven't been certified to have eccentricity $\le 5\rho$ \emph{and} could potentially have small eccentricity must be in $C$, among other properties of $C$ and $B$.

%So the set $W$ is divided into subsets $W_1,\ldots, W_k$ which are the unions of connected components of the far graph, and we call them {\it supercomponents}. 

\subsubsection*{\emph{Analysis Step 3}}
First note that one major difference of this algorithm and the min-radius algorithm is in part b; in the min-radius algorithm we stop whenever we find a good approximate center among the vertices in $V_e$s, but here we can only upper bound the eccentricity of some vertices by $5\rho$ if we find vertices with eccentricity $\le3\rho$ among $V_e$s. 

We first show that if for some vertex $s$ and for some consecutive pair of vertices $(a,b)$ in $V_e$ such that $\eps(a),\eps(b)\le 3\rho$ and $s\in S_b^{\rho}\cap T_a^{\rho}$, then $\eps(s)\le 5\rho$. This is derived by Lemma~\ref{lem:triangle-lemma-ecc} which we state bellow, by the following substitution: let $c=s$, $\gamma_1=\rho,\gamma_2=2\rho$ and $\gamma_3=3\rho$. 
%we cannot bound the eccentricity of 

\begin{lemma}\label{lem:triangle-lemma-ecc}
Consider vertices $b, c$ such that $d(b, c) \leq \gamma_1$, $d(c, b) \leq \gamma_2$ and $\eps(b) \leq \gamma_3$ then $\eps(c) \leq \gamma_3+\max(\gamma_1, \gamma_2)$.
\end{lemma}

\begin{proof}
Consider a vertex $v$, as $\eps(b) \leq \gamma_3$ either $d(v, b) \le \gamma_3$ or $d(b, v) \le \gamma_3$. If $d(v, b) \le \gamma_3$ then $d(v, c) \le d(v, b)+d(b, c) \le \gamma_3+\gamma_1$. Otherwise $d(b, v) \le \gamma_3$ then $d(c, v) \le d(c, b)+d(b, v) \le \gamma_3+\gamma_2$. In both cases $\eps(c) \leq \gamma_3+\max(\gamma_1, \gamma_2)$.
\end{proof}

Now we observe an important property of supercomponents with an analogous proof to that of Observation~\ref{obs:backedge}.

\begin{observation}[Modification of Observation~\ref{obs:backedge}]\label{obs:backedge-ecc}
For every pair of vertices in $v_i \in W_i$ and $v_j \in W_j$ with $i<j$, $d(v_j,v_i)>5\rho$.
\end{observation}

We now prove a modification of Lemma~\ref{lem:closegraph}. This further prunes the vertices that could potentially have small eccentricity.

\begin{lemma}[Modification of Lemma~\ref{lem:closegraph}]
\label{lem:closegraph-ecc}
%close graph lemma
If for a vertex $s \in V$, $\eps(s)\le \rho$ and we haven't yet certified $\eps(s) \le 5\rho$ then for each $i=1,\ldots,k$, either $s\in S_{W_i}$ or $s\in T_{W_i}$.
\end{lemma}
\begin{proof}
Fix $i$ and suppose by way of contradiction that there are nodes $u,v\in W_i$ such that $s\in S_u\cap T_v$ and $\eps(s) \leq \rho$. By Lemma~\ref{lem:SorT-ecc}, $u$ and $v$ must be in different $F_i$'s say $F_u$ and $F_v$. %Now we show that we can assume that $d(u,v)<5\rho$:

 Recall that by the definition of a supercomponent, $f_u$ and $f_v$ are in the same strongly connected component of $G_{close}$. So there is a path $P$ from $f_u$ to $f_v$ in $G_{close}$ such that all of its edges are in $E'$. By Lemma~\ref{lem:SorT-ecc} since $s\in S_u\cap T_v$, we have that $s\in S_{F_u}\cap T_{F_v}$. So there are two consecutive nodes $f_{j}$ and $f_{j'}$ on $P$ (in that order) such that $s\in S_{F_{j}}\cap T_{F_{j'}}$. %Take two nodes $u'\in F_j$ and $v'\in F_{j'}$. Then $d(u',v')<5\rho$, and $c\in S_{u'}\cap T_{v'}$.

%So we assume that $d(u,v)<5\rho$. 
Recall that an edge $e\in E'$ corresponds to a path $P_e$ of length at most $5\rho$ in the original graph. Let $e$ be the edge $(f_j,f_{j'})$ and consider $P_e$ and $V_e$. Since the endpoints of $P_e$ are in $F_j$ and $F_{j'}$ respectively, there exists a pair of vertices $u',v'$ consecutive in $V_e$ (in that order) such that $s\in S_{u'}\cap T_{v'}$. We note that $d(u',v')\leq \rho$.

%old version:
%So we assume that $d(u,v)<5\rho$. 
%Now by Lemma \ref{lem:divideTo_r_sections}, there are at most $9$ nodes $v_1,\ldots,v_9$ on the path from $u$ to $v$ that divide the path into subpaths of length at most $\rho$. So there is some $i\in\{0,\ldots,9\}$ such that $c\in T_{v_i}$ and $c\in S_{v_{i+1}}$, where $v_0=v$ and $v_{10}=u$.

%Now we claim that $v_i$ has min-eccentricity at most $3r$: This is because $d(v_i,c)\le r$ and $d(c,v_i)\le d(c,v_{i+1})+d(v_{i+1},v_i) \le 2\rho$. So considering any vertex $w\in V$, if $d(c,v)\le R$, then $d(v_i,v)\le d(v_i,c)+d(c,v) \le 2\rho$. If $d(v,c)\le R$, then $d(v,v_i)\le d(v,c)+d(c,v_i)\le 3r$.

%Using Lemma~\ref{} with $b=s,\gamma_3=\rho,c=v',\gamma_1=2\rho,\gamma_2=\rho$, we get that $\eps(v')\leq \rho+\max\{2\rho,\rho\}=3\rho,

%Now we claim that $\eps(v')\leq 3\rho$. 
Recall that we assumed that $\eps(s)\leq \rho.$
Note as well that $d(v',s)\le \rho$ and $d(s,v')\le d(s,u')+d(u',v') \le 2\rho$. Then, using Lemma~\ref{lem:triangle-lemma-ecc} with $b=s,\gamma_3=\rho,c=v',\gamma_1=2\rho,\gamma_2=\rho$, we get that $\eps(v')\leq \rho+\max\{2\rho,\rho\}=3\rho$.
A symmetric argument holds for $u'$, giving $\eps(u'),\eps(v')\leq 3\rho$. In this case, the algorithm would have already marked $s$ in step 3b as it is in the intersection of $S^{\rho}_{u'} \cup T^{\rho}_{v'}$.
\end{proof}

%Consider an arbitrary vertex $w\in V$. Either $d(s,w)\le \rho$ or $d(w,s)\le \rho$. If $d(s,w)\le \rho$ then $d(v',w)\le d(v',s)+d(s,w) \le 2\rho$. If $d(w,s)\le \rho$, then $d(w,v')\le d(w,s)+d(s,v')\le 3\rho$. 
%A symmetric argument holds for $u'$. In this case, the algorithm would have already marked $s$ in step 3b as it is in the intersection of $S^{\rho}_{u'} \cup T^{\rho}_{v'}$.
%\end{proof}

We now prove that for vertices $s$ which have small min-eccentricity and have not been certified as such, $s \in C$. The proof is analogous to that of Lemma~\ref{lem:propC}.

\begin{lemma}[Modification of Lemma~\ref{lem:propC}]
\label{lem:propC-ecc}
If for a vertex $s \in V$, $\eps(s)\le \rho$ and we haven't yet certified $\eps(s) \le 5\rho$ then $s\in C$.
\end{lemma}

Now we prove that the vertices in $B$ fit into the DAG structure in a similar but weaker sense than the vertices in $C$. The proofs are analogous to those of Lemma~\ref{lem:propB} and Corollary~\ref{cor:Bproperty_extension}.

%\begin{enumerate}[label=\roman*]
  %  \item \label{twofarcomp}
     
    %For any two nodes $w_i\in W_i$ and $w_j\in W_j$ for some $i<j$, since $w_i$ and $w_j$ are in different connected components in the far graph, there is a path of length at most $2\rho$ from $w_i$ to $w_j$.
  %  \item \label{backedgeprop} For each $i<j$, $S_{W_i}^{2\rho}\cap T_{W_j}^{2\rho} = \emptyset$, i.e. there is no vertex $v$, where $d(W_j,v)\le 2\rho $ and $d(v,W_i)\le 2\rho$, because otherwise there is a path of length at most $4r$ from $W_j$ to $W_i$, contradicting the DAG ordering. 
%\end{enumerate}

\begin{lemma}[Modification of Lemma~\ref{lem:propB}]
\label{lem:propB-ecc}
    Consider a node $v\in B_i$. Then for all $z\le i$ except for at most two values, we have $v\in T_{W_z}^{2\rho}$. And for all $z\ge  i$ except for at most two values, we have $v\in S_{W_z}^{2\rho}$.

\end{lemma}
\begin{corollary}[Modification of Corollary~\ref{cor:Bproperty_extension}]
\label{cor:Bproperty_extension-ecc}
Lemma~\ref{lem:propB-ecc} is true for all $v\in B_i\cup C_i\cup W_i$. Moreover for all $v\in B_i\cup C_i\cup W_i$, we have $v\in T_{W_{i-1}}^{2\rho}$.
\end{corollary}

\subsubsection*{Algorithm Step 4: Partial search}

From each of the potential vertices with small min-eccentricity in $C$, we will run Dijkstra's algorithm on a small subgraph of $G$. $G_i$ and $ U_i$ are defined as in the min-radius algorithm. Define $\bar{C}_i$ to be the set of nodes $v \in C_i$ such that $v$ is within min-distance $\rho$ from all vertices in $W$ (we know this set of nodes because we have already run Dijkstra's algorithm from and to every vertex in $W$). From each node $v\in \bar{C}_i$ run Dijkstra's algorithm from and to $v$ with respect to the graph $G_i$. If $v$ is within min-distance $\rho$ from all nodes in $U_i$, we will show that this certifies that $\eps(s)\leq 3\rho$ and otherwise $\eps(s) > \rho$.

\subsubsection*{\emph{Analysis Step 4}} 
%This step of the algorithm takes $O(m\sqrt{n})$ time. (amortized argument TODO).

The following two claims prove that our algorithm for vertices $s \in C$ either certifies that $\eps(s)>\rho$ or $\eps(s)\leq 3\rho$. The proofs are analogous to those of Claim~\ref{claim:sp_inside} and Claim~\ref{claim:approx_center_valid}.

%The correctness of the entire algorithm follows from Claims~\ref{claim:sp_inside} and \ref{claim:approx_center_valid}:

\begin{claim}[Modification of Claim~\ref{claim:sp_inside}]
\label{claim:sp_inside-ecc}
 If $s\in C_i$ and $\eps(s)\le \rho$, then for all $u\in U_i$, the min-distance between $c$ and $u$ with respect to $G_i$ is at most $\rho$. 
\end{claim}
\begin{claim}[Modification of Claim~\ref{claim:approx_center_valid}]
\label{claim:approx_center_valid-ecc}
If a vertex $s$ is within min-distance $\rho$ from all vertices in $U_i$ in $G_i$, then $\eps(s)\leq 3\rho$.
\end{claim}
For all the vertices for which we haven't certified either $\eps(s) \le 3\rho$ or $\eps(s) \le 5\rho$ we know that $\eps(s) > \rho$ and can certify that.

%{\it rest of the Analysis Step 3.} 
%If $R\le r$, by Claim \ref{claim:sp_inside} we know that the algorithm will output some node, and by Claim \ref{claim:approx_center_valid} we see this node has min-eccentricity at most $3r$. So if the algorithm doesn't return any node, it means that $R>r$.

\end{proof}

The runtime is $\tO(m\sqrt{n})$ with analogous runtime analysis to that of the min-radius algorithm.

\subsection{$(3+\delta)$-approximation for unweighted graphs}

In this part we show that given an unweighted graph, by a slight modification of the min-eccentricity algorithm in Theorem \ref{thm:eccntricities}, we are able to improve the approximation factor of the min-eccentricity problem to match that of the min-radius problem, namely we present a $(3+\delta)$-approximation algorithm for every $\delta>0$.

%Below we describe a slight modification of the algorithm in Theorem \ref{thm:eccntricities} for unweighted graphs which allows us to get a $(3+\delta)$-approximation for every $\delta > 0$.

\begin{theorem}
\label{thm:eccntricities-unweighted}
For any constant $\delta$ with $1>\delta>0$, there is an $\tilde{O}(m\sqrt n/\delta^2)$ time randomized algorithm, that given a directed unweighted graph $G=(V,E)$, can output an estimate $\eps'(s)$ for every vertex $s \in V$ such that $\eps(s)\le \eps'(s)\le (3+\delta)\eps(s)$ with high probability, where $\eps(s)$ is the min-eccentricity of the vertex $s$ in $G$. 
\end{theorem}
\begin{proof}
There are only two parts of the algorithm in Theorem \ref{thm:eccntricities} that change: 

{\bf (1)} Letting $\delta'=\delta/5$, in each run of the algorithm, for each vertex $s$, we certify that either $\eps(s)>\rho$ or $\eps(s)\leq (3+\delta')\rho$ (instead of $\eps(s)\leq 5\rho)$. The subsequent changes follow naturally: We start from $\rho = 1$ and we run the algorithm and increase $\rho$ by a factor of $(1+\delta')$. We call the vertices for which we have certified $\eps(s)\leq (3+\delta')\rho$ for earlier values of $\rho$ as \emph{marked}, and if for an unmarked vertex $s$ the output of the algorithm is $\eps(s)\leq (3+\delta')\rho$, then we let $\eps'(s)=(3+\delta')\rho$. If for some value of $\rho$ and for some vertex $s$ we have $\eps(s)\le (3+\delta')\rho$ and $s$ was unmarked, then from the previous run of the algorithm, we know that $\eps(s) >\rho/(1+\delta')$. So for $\eps'(s)=(3+\delta')\rho$, we have $\eps'(s) \geq \eps(s)$ and $\eps'(s)\le (3+\delta')(1+\delta')\eps(s)=(3+\delta)\eps(s)$.

{\bf (2)} In step 3, part b of the algorithm  (Additional Dijkstra), recall that each edge $e\in E'$ is a path of length at most $5\rho$ in $G$. Now instead of dividing each $e$ into at most $9$ subpaths of length at most $\rho$, we divide it into subpaths of length at most $\delta'\rho/2\ge 1$ using at most $20/\delta'-1=O(1/\delta')$ vertices which we call $V_e$. The rest of this step follows naturally: We run Dijkstra from and to each $v\in V_e$, so we know that whether $\eps(v)>\rho$ or $\eps(v)<(2+\delta'/2)\rho$. For every consecutive pair of vertices $(a, b)$ in $V_e$ over all $e$ with $\eps(a), \eps(b) \leq (2+\delta'/2)\rho$ we certify for all $s \in S^{\rho}_b \cap T^{\rho}_a$ that $\eps(s) \le (3+\delta')\rho$. This is indeed true by Lemma~\ref{lem:triangle-lemma-ecc} (in the statement of the lemma, let $c=s$, $\gamma_1=\delta\rho/2,\gamma_2=(1+\delta/2)\rho$ and $\gamma_3=(2+\delta'/2)\rho$). 

First note that by this change the number of vertices that we do Dijkstra from/to in step 3(b) of the algorithm is now $O(|W|/\delta')=\tilde{O}(\sqrt{n}/\delta')=\tilde{O}(\sqrt{n}/\delta)$ (see runtime analysis of step 3(b) in Theorem \ref{thm:radius}). The runtime of the other steps are not changed, so the overall runtime of the algorithm is $\tilde{O}(m\sqrt{n}/\delta^2)$.

The main issue in the min-eccentricity algorithm that didn't allow us to get a $(3+\delta')$ approximation is that we could have potentially big weighted edges, and that didn't let us divide $5\rho$-length paths into smaller parts. The analysis of this part is due to Lemma \ref{lem:closegraph-ecc}, which is modified as in Lemma \ref{lem:closegraph-ecc-unweighted}.

%We fix a value $\rho$ and our algorithm certifies for each $s \in V$ that either $\eps(s)>\rho$ or $\eps(s)\leq (3+\delta')\rho$ with high probability. Starting from $\rho=1$, we will run the algorithm and increase $\rho$ for each run. We will call the vertices for which we have certified $\eps(s)\leq (3+\delta')\rho$ for earlier values of $\rho$ as \emph{marked}. Let $\delta' = \delta/3$. Starting from $\rho=1$, we run the algorithm. If the output of the algorithm is that $\eps(s)\le (3+\delta')\rho$ and $s$ was unmarked, then we will mark $s$ and set $\eps'(s) = (3+\delta')\rho$. Then, we run the algorithm with the new value $\rho_{new}=(1+\delta')\rho$. Since $\eps(s)\le poly(n)$ for all $s\in V$, this contributes a multiplicative factor of $\log_{1+\delta'}n=\tilde{O}(1/\delta)$ to the total runtime. Suppose that for some value of $\rho$ and for some vertex $s$ we have $\eps(s)\le (3+\delta')\rho$ and $s$ was unmarked. So from the previous run of the algorithm, we know that $\eps(s) >\rho/(1+\delta')$. Then for $\eps'(s)=(3+\delta')\rho$, we have $\eps'(s) \geq \eps(s)$ and $\eps'(s)\le (3+\delta')(1+\delta')\eps(s)=(3+\delta)\eps(s)$, which means that $\eps'(s)$ is a $(3+\delta)$-approximation of $\eps(s)$. After running the whole algorithm for this value of $\rho$ we will also mark all such vertices $s$. Now we present the algorithm.
\end{proof}

\begin{lemma}[Modification of Lemma~\ref{lem:closegraph-ecc}]
\label{lem:closegraph-ecc-unweighted}
%close graph lemma
If for a vertex $s \in V$, $\eps(s)\le \rho$ and we haven't yet certified $\eps(s) \le (3+\delta')\rho$ then for each $i=1,\ldots,k$, either $s\in S_{W_i}$ or $s\in T_{W_i}$.
\end{lemma}
\begin{proof}
The proof is similar to that of Lemma \ref{lem:closegraph-ecc}, with a change at the end of the argument because of our finer division of paths. Fix $i$ and suppose by way of contradiction that there are nodes $u,v\in W_i$ such that $s\in S_u\cap T_v$ and $\eps(s) \leq \rho$. Similar to Lemma \ref{lem:closegraph-ecc}, we can assume that there are two vertices $u',v'$ that we have done Dijkstra from such that $s\in S_{u'}\cap T_{v'}$ and $d(u',v')\le \rho\delta'/2$.

Now we claim that $\eps(v')\leq (2+\delta'/2)\rho$. Note that $d(v',s)\le \rho$ and $d(s,v')\le d(s,u')+d(u',v') \le (1+\delta'/2)\rho$. Consider an arbitrary vertex $w\in V$. Either $d(s,w)\le \rho$ or $d(w,s)\le \rho$. If $d(s,w)\le \rho$ then $d(v',w)\le d(v',s)+d(s,w) \le 2\rho$. If $d(w,s)\le \rho$, then $d(w,v')\le d(w,s)+d(s,v')\le (2+\delta'/2)\rho$. A symmetric argument holds for $u'$. In this case, the algorithm would have already marked $s$ in step 3b as it is in the intersection of $S^{\rho}_{u'} \cap T^{\rho}_{v'}$.
\end{proof}

%The runtime analysis is the same except for step 3b) where we now run Dijkstra from $O(\sqrt{n}/\delta'') = O(\sqrt{n}/\delta)$ vertices instead of $O(\sqrt{n})$ vertices. So the algorithm runs in time $\tO(m\sqrt{n}/\delta)$.

%todo add to runtime the fact that C_is are small. 

\section*{Acknowledgements} 
\noindent The authors would like to thank the members of the MIT course 6.S078 open problem sessions, especially Thuy-Duong Vuong, Robin Hui, and Ali Vakilian. These sessions were organized by Erik Demaine, Ryan Williams, and Virginia Vassilevska Williams, and used the collaboration software Coauthor, created by Erik Demaine.
%V:since we aren't including our names, we should also not include this for now
\bibliographystyle{plain}
\bibliography{references}

\appendix
\section{Appendix}

\begin{lemma}
\label{lem:divideTo_r_sections}
Given a weighted graph $G$ and a path $P$ in $G$ from $v$ to $u$ of length at most $zr$ for some integers $z$ and $r$, one can find in $O(|P|)$ time vertices $v_1,\ldots, v_{z'}$ such that $z'\le 2z-1$ and they divide $P$ into subpaths of length at most $r$ if there are no edges of weight more than $r$ on the path. Equivalently, $|P_{v_iv_{i+1}}|\le r$, for $i=0,\ldots,z$, where $v_0=v, v_{z'+1}=u$ and $P_{v_iv_{i+1}}$ is the part of the path $P$ between $v_i$ and $v_{i+1}$. 
\end{lemma}

\begin{proof}
Start from $v_0=v$ and go through the path until the last vertex $w$ such that $d(v,w)\le r$ but $d(v,w')>r$ where $w'$ is the node right after $w$ on the path. Note that since there are no edges of weight more than $r$, such $w$ exists. Let $v_1=w$. Starting from $v_1$, we can do the same and find all vertices $v_2,\ldots, v_{z'}$. It is remained to prove that $z'<2z$. By the definition of $v_1$, we know that $d(v_0,v_2)>r$. Similarly, we can argue that $d(v_i,v_{i+2})>r$ for all $i=0,\ldots, z'-1$. So $d(v_0,v_{2i})>ir$. Since $|P|\le zr$, we have $z'\le 2z-1$. We went through the vertices of $P$ once, so the running time is linear in terms of the length of the path.  
\end{proof}

\begin{lemma}
\label{lem:SCC}
There is an algorithm that given a strongly connected graph $H=(V,E)$, outputs in $O(|E|)$ time a subset $E'\subseteq E$ of size at most $2(|V|-1)$ such that $H'=(V,E')$ is strongly connected.
\end{lemma}
\begin{proof}
For any vertex $v$ do a BFS to and from $v$ and denote by $E'$ the union of edges in the two computed BFS trees. $H'=(V,E')$ is strongly connected as for every ordered pair of vertices $(a, b)$ we can go from $a$ to $b$ by following the path $a \to v \to b$. It is clear that since $E'$ is the union of two trees, $|E'|\leq 2(|V|-1)$.

%The algorithm is by recursion: If $|E|\le 2|V|$ output $E'=E$. Otherwise since $|E|\ge|V|$, the graph has at least one cycle. Take an arbitrary cycle $C$ and contract it into one node. Equivalently, let $\tilde{H}=(\tilde{V}, \tilde{E})$ be the graph obtained by replacing all vertices of $C$ by a single node $a$. All the edges in $H$ with one endpoint in $C$ have that endpoint replaced by $a$ in $\tilde{H}$ and all the edges within vertices in $C$ in $H$ are removed in $\tilde{H}$. Then recurse on $\tilde{H}$. Let $\tilde{E}'$ be the output of the recursion. Let $E'=\tilde{E}'\cup C$.

%We prove that $|E'|\le 2|V|$ by induction on $|E|$. If $|E|\le 2|V|$, then we are done. Suppose $|E|>2|V|$. By the induction hypothesis, $|\tilde{E}'|\le 2|\tilde{V}| \le 2(|V|-|C|+1)$. So $|E'|=|\tilde{E}'|+|C|\le 2|V|-|C|+2 \le 2|V|$.

%For the time analysis, note that finding a cycle takes $O(|V|)$ time and updating the graph takes $O(|C|\cdot |V|) $ time where $|C|$ is the cycle we find. Since the cycles of each step are disjoint, the runtime is $O(|V|^2)$.
\end{proof}

\end{document}